\documentclass{article}

\usepackage[preprint]{neurips_2022}

\usepackage[utf8]{inputenc} 
\usepackage[T1]{fontenc}    
\usepackage{hyperref}       
\usepackage{url}            
\usepackage{booktabs}       
\usepackage{amsfonts}       
\usepackage{nicefrac}       
\usepackage{microtype}      
\usepackage{xcolor}         

\usepackage{paralist, amsmath, amsthm, amssymb, color, graphicx, hyperref}
\DeclareGraphicsExtensions{.jpg,.pdf,.eps}
\usepackage{amssymb}
\usepackage{todonotes}
\usepackage{setspace}
\usepackage{mathtools}
\usepackage{color}
\usepackage{cancel}
\usepackage{algorithm,algorithmic}
\usepackage{verbatim}

\usepackage{stmaryrd}

\usepackage{xfrac}

\usepackage{mathrsfs}
\usepackage{amscd}
\usepackage{dsfont}
\usepackage{tikz}

\usepackage{enumerate}

\usepackage{thmtools}
\usepackage{thm-restate}
\theoremstyle{plain}
\newtheorem{thm}{Theorem}[section] 

\newtheorem{prop}[thm]{Proposition}

\newtheorem{lemma}[thm]{Lemma}

\theoremstyle{definition}
\newtheorem{defn}[thm]{Definition} 

\newtheorem{example}[thm]{Example}
\newtheorem{remark}[thm]{Remark}

\usepackage{listings}
\usepackage{color} 
\definecolor{mygreen}{RGB}{28,172,0} 
\definecolor{mylilas}{RGB}{170,55,241}
\definecolor{mygray}{gray}{0.95}

\newcommand{\mscr}[1]{\mathscr{#1}}

\newcommand{\ZZ}{\mathbb{Z}}
\newcommand{\RR}{\mathbb{R}}
\newcommand{\NN}{\mathbb{N}}

\newcommand{\QQ}{\mathbb{Q}}
\newcommand{\EE}{\mathbb{E}}

\newcommand{\de}{\delta}
\newcommand{\ep}{\epsilon}

\usepackage{mathtools}

\renewcommand{\l}{\left}
\renewcommand{\r}{\right}

\newcommand{\defeq}{\vcentcolon=}

\DeclareMathOperator{\argmax}{argmax}
\DeclareMathOperator{\TV}{TV}

\newcommand{\iid}{\overset{\text{iid}}{\sim}}

\DeclareMathOperator*{\esssup}{esssup}

\newcommand{\djs}[1]{{\color{blue} Jinshuo: #1}}

\title{Log-Concave and Multivariate Canonical Noise Distributions for Differential Privacy}

\author{%
Jordan A. Awan\\
Department of Statistics\\
Purdue University\\
\texttt{jawan@purdue.edu}\\
\And
Jinshuo Dong\\
Department of Computer Science\\
Northwestern University and IDEAL\thanks{The Institute for Data, Econometrics, Algorithms, and Learning}\\
\texttt{jinshuo@northwestern.edu}
}

\begin{document}

\maketitle

\begin{abstract}
 A canonical noise distribution (CND) is an additive mechanism designed to satisfy $f$-differential privacy ($f$-DP), without any wasted privacy budget. $f$-DP is a hypothesis testing-based formulation of privacy phrased in terms of \emph{tradeoff functions}, which captures the difficulty of a hypothesis test. In this paper, we consider the existence and construction of both log-concave CNDs and multivariate CNDs. Log-concave distributions are important to ensure that higher outputs of the mechanism correspond to higher input values, whereas multivariate noise distributions are important to ensure that a joint release of multiple outputs has a tight privacy characterization. We show that the existence and construction of CNDs for both types of problems is related to whether the tradeoff function can be decomposed by functional composition (related to group privacy) or mechanism composition. In particular, we show that pure $\epsilon$-DP cannot be decomposed in either way and that there is neither a log-concave CND nor any multivariate CND for $\epsilon$-DP. On the other hand, we show that Gaussian-DP, $(0,\delta)$-DP, and Laplace-DP each have both log-concave and multivariate CNDs. 
\end{abstract}

\section{Introduction}
Differential privacy (DP), proposed by \citet{dwork2006calibrating}, is the state-of-the-art framework in formal privacy protection and is being implemented by tech companies, government agencies, and academic institutions. Over time, the DP community has developed many new DP mechanisms as well as new frameworks. Recently, $f$-DP (\citet{dong2022gaussian} was proposed as a generalization of DP, allowing for tight calculations of group privacy, composition, subsampling, and post-processing. It was shown in \citet{dong2022gaussian} that $f$-DP is provably the tightest version of DP that respects the post-processing property of DP. In particular, $f$-DP can be losslessly converted to R\'enyi-DP (or any $f$-divergence version of DP) as well as $(\ep,\de)$-DP, but not vice-versa \citep{dong2022gaussian}. Furthermore, $f$-DP is equivalent (can be losslessly converted back and forth) to the privacy profile \citep{balle2018privacy,balle2020privacy} and the privacy loss random variables \citep{sommer2019privacy,zhu2022optimal}.

$f$-DP is defined in terms of a \emph{tradeoff function} or \emph{receiver operator curve  (ROC)}, which encapsulates the difficulty of conducting a hypothesis test between two distributions. If $f=T(P,Q)$ is the tradeoff function for testing between the distributions $P$ and $Q$, then if a mechanism $M$ satisfies $f$-DP, this means that given the output of $M$ when run on one of two adjacent databases, it is at least as hard to determine which database was used, as it is to test between $P$ and $Q$. 



While $f$-DP has the many desirable theoretical properties listed above in its favor, there are  limited techniques for working with $f$-DP, and few constructive mechanisms for an arbitrary $f$-DP guarantee. A notable exception is a \emph{canonical noise distribution} (CND) from the recent paper \citet{awan2021canonical}, which builds a one-dimensional additive noise mechanism designed to exactly satisfy $f$-DP, with no wasted privacy budget. Along with the intuitive idea that a CND is optimal in that it optimizes the privacy loss budget, \citet{awan2021canonical} showed that CNDs are crucial to the construction of optimal DP hypothesis tests and free DP $p$-values. However, the CND construction given in \citet{awan2021canonical} does not result in a smooth distribution, and in particular is not \emph{log-concave}. 
Log-concavity is a desirable property because it implies that the distribution has a monotone likelihood ratio; this means that higher observed values are always more likely to have come from a higher input value than a lower one. Log-concavity thus makes the DP output much more interpretable, easily analyzed, and also has makes the calculation of the privacy cost simpler \citep{dong2021central}. Furthermore, the results of \citet{awan2021canonical} are limited to 1-dimensional distributions.

In this paper, we develop new properties of CNDs and $f$-DP, motivated by the following two questions, 
\begin{center}
    1. Can we construct log-concave CNDs?\qquad
    2. Can we construct multivariate CNDs?
\end{center}

{\bf Our Contributions} 
The existence of both log-concave 1-dimensional CNDs and multivariate CNDs are intricately linked with properties related to group privacy and mechanism composition. 
Two highly desirable properties of a tradeoff function are \emph{infinite divisibility} and \emph{infinite decomposability}, meaning that the tradeoff function can be exactly achieved by $n$-fold group privacy or $n$-fold mechanism composition, respectively. We prove that a tradeoff function has a log-concave CND if and only if the tradeoff function is infinitely divisible, and give a construction for the \emph{unique} log-concave CND in this case. We also show that if a tradeoff function is either infinitely divisible or decomposable, then we can construct a multivariate CND. 
 
 Along with the positive results listed above, we also include impossibility results. In particular, $(\ep,0)$-DP is neither divisible nor decomposable, and in fact has neither a log-concave CND nor any multivariate CND. In contrast to $(\ep,0)$-DP, two families that satisfy both infinite divisibility and infinite decomposability are $\mu$-GDP and $(0,\de)$-DP. While $(0,\de)$-DP  has limited applicability due to its weak protection for events with small probability, $\mu$-GDP and related DP definitions (such as zero concentrated DP) have been gaining popularity. The results of this paper provide a new perspective supporting the adoption of GDP as the default privacy measure instead of $(\ep,0)$-DP. 

    {\bf Organization } 
    In Section \ref{s:background}, we review concepts in $f$-DP and canonical noise distributions. In Section \ref{s:1d}, we study 1-dimensional CNDs. In Section \ref{s:pure}, we prove that the Tulap distribution is the \emph{unique} CND for $(\ep,0)$-DP. In Section \ref{s:logconcave}, we propose the concept of infinite divisibility and prove that a tradeoff function has a log-concave CND if and only if it is infinitely divisible; we also give a construction to produce the log-concave CND from a family of infinitely divisible tradeoff functions. We prove that piece-wise linear tradeoff functions are generally not infinitely divisible in Section \ref{s:piece}, and in particular $(\ep,0)$-DP and several related tradeoff functions do not have log-concave CNDs. In Section \ref{s:multi}, we propose a multivariate extension of CND. We give two general constructions of multivariate CNDs in Section \ref{s:construction} depending on whether a tradeoff function is decomposable or infinitely divisible. We give several examples of multivariate CNDs in Sections \ref{s:gdp}-\ref{s:laplace} for Gaussian DP, $(0,\delta)$-DP, $(\ep,\de)$-DP, and Laplace-DP. In Section \ref{s:nopure}, we show that there is no multivariate CND for $(\ep,0)$-DP, which implies that $(\ep,0)$-DP is not decomposable. We conclude with discussion in Section \ref{s:discussion}. Proofs and technical details are found in the Appendix.

    {\bf Related Work } While there are many complex DP mechanisms, many use the fundamental building block of additive mechanisms (e.g., functional mechanism \citep{zhang2012functional}, objective perturbation \citep{chaudhuri2011differentially,kifer2012private}, stochastic gradient descent \citep{abadi2016deep}, and the sparse vector technique \citep{dwork2009complexity,zhu2020improving}, to name a few). There have been many different additive mechanisms proposed in the literature, for different privacy purposes. We highlight the works that show some optimality property for the proposed noise distributions.  This work is most directly building off of \citet{awan2021canonical}, who proposed the concept of canonical noise distributions as a method of quantifying what it means to fully use the privacy budget. There are also other works, which derive optimal mechanisms with respect to other metrics. 
   \citet{ghosh2012universally} showed that a discrete Laplace distribution is the universal utility maximizer for a general class of utility functions in pure-DP. 
   \citet{geng2015optimal} proposed the staircase mechanism which they showed optimizes the $\ell_1$ or $\ell_2$ error for pure-DP. For $(\ep,\de)$-DP, \citet{geng2015approx} showed that either the staircase or a uniform distribution can achieve the optimal rate in terms of $\ell_1$ and $\ell_2$ error. \citet{steinke2016between} showed that the $\ell_\infty$-mechanisms is rate optimal when measuring utility in terms of  $\ell_\infty$ error. 
\citet{awan2020structure} derive optimal mechanisms among the class of $K$-Norm Mechanisms, proposed by \citet{hardt2010geometry}, in terms of various scale-independent measures, for a fixed statistic and sample size.


\section{Differential privacy basics}\label{s:background}


Differential privacy ensures that given the output of a private mechanism, it is difficult for an adversary to determine whether an individual is present in the database or not. 
To satisfy DP, a privacy expert employs a \emph{mechanism} $M$, which is a set of probability distributions $M_D$ on a common space $\mscr Y$, indexed by possible databases $D\in \mscr D$.  
Let $d(D,D')$ be an integer-valued metric on the space of databases $\mscr D$, which represents the number of entries that $D$ and $D'$ differ in. We call $D$ and $D'$ \emph{adjacent} if $d(D,D')\leq 1$.  
While there are now many variants of DP, they all center around the idea that given a randomized algorithm $M$, for any two adjacent databases $D$, $D'$, the distributions of $M(D)$ and $M(D')$ should be ``similar.'' While many DP variants measure similarity in terms of divergences, $f$-DP formalizes similarity in terms of hypothesis tests.   
Intuitively, for two adjacent databases $D$ and $D'$, a mechanism $M$ satisfies $f$-DP if given the output of $M$, it is difficult to determine whether the original database was $D$ or $D'$. This is formalized in terms \emph{tradeoff functions}.


For two distributions $P$ and $Q$, the \emph{tradeoff function} (or ROC) between $P$ and $Q$ is $T(P,Q):[0,1]\rightarrow [0,1]$, where $T(P,Q)(\alpha)=\inf \{1-\EE_Q \phi \mid \EE_P(\phi)\geq 1-\alpha\}$, where the infinimum is over all measurable tests $\phi$. The tradeoff function returns the optimal type II error for testing $H_0=P$ versus $H_1=Q$ at specificity (one minus type I error) $\alpha$, and captures the difficulty of distinguishing between $P$ and $Q$.
\footnote{In \citet{dong2022gaussian}, the tradeoff function was originally defined as a function of type I error. Our choice to flip the tradeoff function along the $x$-axis is for mathematical convenience. The ROC function is usually defined as the power (one minus type II error) as a function of type I error.} 
A function $f:[0,1]\rightarrow [0,1]$ is a tradeoff function if and only if $f$ is convex, continuous, non-decreasing, and $f(x) \leq x$ for all $x \in [0,1]$ \citep[Proposition 2.2]{dong2022gaussian}.  We say that a tradeoff function $f$ is \emph{nontrivial} if $f(\alpha)<\alpha$ for some $\alpha\in (0,1)$.

\begin{defn}[$f$-DP: \citealp{dong2022gaussian}]\label{def:fDP}
 Let $f$ be a tradeoff function. A mechanism $M$ satisfies $f$-DP if $T(M(D),M(D'))\geq f,$  for all $D,D'\in \mscr D$ which satisfy $d(D,D')\leq 1$.
\end{defn}

Intuitively, a mechanism satisfies $f$-DP, where $f=T(P,Q)$, if testing $H_0:M(D)$ versus $H_1: M(D')$ is at least as hard as testing $H_0:P$ versus $H_1:Q$. Without loss of generality we can assume that $f$ is \emph{symmetric}, meaning that if $f=T(P,Q)$, then $f=T(Q,P)$. This is due to the fact that adjacency of databases is a symmetric relation \citep[Proposition 2.4]{dong2022gaussian}. So, we limit the focus of this paper on symmetric tradeoff functions.

A key property of differential privacy is that it also implies privacy guarantees for groups.  \citet[Theorem 2.14]{dong2022gaussian} showed that if a mechanism is $f$-DP, then it satisfies $f^{\circ k}$-DP, when the adjacency measure is changed to allow for a difference in $k$ entries (where $f^{\circ k}$ means the functional composition of $f$ with itself, $k$ times). We call this \emph{group privacy}, which is a central topic in differential privacy.  
Note that the bound $f^{\circ k}$ is not necessarily the tightest privacy guarantee for a particular mechanism. 

\emph{Mechanism Composition} quantifies the cumulative privacy cost of the output of $k$  mechanisms. To express the tradeoff function resulting from composition, \citet{dong2022gaussian} proposed the \emph{tensor product} of tradeoff functions: if $f=T(P,Q)$ and $g=(P',Q')$, then $f\otimes g \defeq T(P\times P',Q\times Q')$, which they show is well defined, commutative, and associative. They prove that if we have $k$ 
mechanisms $M_1,\ldots, M_k$, which each satisfy $f_1$-DP, $f_2$-DP,$\ldots, f_k$-DP respectively, then the composition $(M_1,\ldots, M_k)$ satisfies $f_1\otimes\cdots\otimes f_k$-DP (see \citet[Theorem 3.2]{dong2022gaussian} for a more precise statement).

The traditional framework of $(\ep,\de)$-DP is a subclass of $f$-DP: Let $\ep\geq 0$ and $\de\in [0,1]$. A mechanism satisfies $(\ep,\de)$-DP if it satisfies $f_{\ep,\de}$-DP, where $f_{\ep,\de}(\alpha) =\max\{0,1-\de-e^{\ep}+e^{\ep}\alpha,\exp(-\ep)(\alpha-\de)\}$. An important special case is $(\ep,0)$-DP, which was the original definition of DP. 

Another popular subclass is Gaussian-DP (GDP): For $\mu\geq 0$, a mechanism satisfies $\mu$-GDP if it satisfies $G_\mu$-DP, where $G_\mu = T(N(0,1), N(\mu,1))$. Gaussian-DP was proposed in \citet{dong2022gaussian} and has several desirable properties, such as being closed under group privacy and closed under composition. \citet{dong2022gaussian} also established a central limit theorem for tradeoff functions as the number of compositions approaches infinity, showing that under general assumptions the tradeoff function of the composed mechanisms approaches $G_\mu$ for some $\mu$. 

\subsection{Canonical noise distributions}
To satisfy DP, additive mechanisms must introduce noise proportional to the \emph{sensitivity} of the statistic of interest. Let $\lVert \cdot \rVert$ be a norm on $\RR^d$. A statistic $S:\mscr D\rightarrow \RR^d$ has $\lVert \cdot \rVert$-\emph{sensitivity} $\Delta>0$ if $\lVert S(D)-S(D')\rVert \leq \Delta$ for all $d(D,D')\leq 1$. When $d=1$, we use $|\cdot|$ as the default norm. Any additive mechanism, which releases $S(D)+\Delta N$, satisfies $f$-DP if  $T(N,N+v)\geq f$ for all $\lVert v \rVert\leq 1$. The concept \emph{canonical noise distribution} (CND) was proposed by \citet{awan2021canonical} to capture when an additive mechanism satisfies $f$-DP, and ``fully uses the privacy budget.''

\begin{defn}[Canonical noise distribution: \citet{awan2021canonical}]\label{def:CND}
  Let $f$ be a symmetric tradeoff function. A continuous random variable $N$ with cumulative distribution function (cdf) $F$ is a \emph{canonical noise distribution} (CND) for $f$ if 
  \begin{enumerate}
      \item 
      For any $m\in [0,1]$, $T(N,N+m)\geq f$,
      \item $f(\alpha)=T(N,N+1)(\alpha)$ for all $\alpha \in (0,1)$,
      \item $T(N,N+1)(\alpha) = F(F^{-1}(\alpha)-1)$ for all $\alpha \in (0,1)$,
      \item $F(x) = 1-F(-x)$ for all $x\in \RR$; that is, $N$ is symmetric about zero.
  \end{enumerate}
\end{defn}
In Definition \ref{def:CND}, property 1 ensures that the additive mechanism using a CND satisfies $f$-DP, property 2 ensures that the privacy guarantee is tight, property 3 gives a closed form for the tradeoff function in terms of the CND's cdf, which is equivalent to enforcing a monotone likelihood ratio property, and property 4 imposes symmetry which is mostly for convenience.

An important property of CNDs is that they satisfy the following recurrence relation:
\begin{lemma}[\citet{awan2021canonical}]\label{lem:recurrence}
  Let $f$ be a symmetric nontrivial tradeoff function and let $F$ be a CND for $f$. Then $F(x)=1-f(1-F(x-1))$ when $F(x-1)>0$ and $F(x)=f(F(x+1))$ when $F(x+1)<1$. 
\end{lemma}
In \citet{awan2021canonical}, they showed that the above recurrence relation can be used to construct a CND for any nontrivial symmetric tradeoff function.

\begin{prop}[CND construction: \citet{awan2021canonical}]\label{prop:CNDsynthetic}
  Let $f$ be a symmetric nontrivial tradeoff function, and let $c\in [0,1]$ be the solution to $f(1-c)=c$. We define $F_f:\RR\rightarrow \RR$ as 
  \[ F_f(x) = \begin{cases}
  f(F_f(x+1))&x<-1/2\\
  c(1/2-x) + (1-c)(x+1/2)&-1/2\leq x\leq 1/2\\
  1-f(1-F_f(x-1))&x>1/2.\\
  \end{cases}\]
  Then $N\sim F_f$ is a canonical noise distribution for $f$.
\end{prop}
While Proposition \ref{prop:CNDsynthetic} gives a general construction of a CND for an arbitrary $f$, the resulting distribution is generally not smooth or log-concave. \citet{awan2021canonical} showed that in the case of $G_\mu$, this construction does not recover the Gaussian distribution, which is the log-concave CND.

\section{One-dimensional CNDs}\label{s:1d}
In this section, we expand on the results of \citet{awan2021canonical}, by producing new results for one-dimensional CNDs. In Section \ref{s:pure}, we show that the Tulap distribution is the \emph{unique} CND for $(\ep,0)$-DP. In Section \ref{s:logconcave}, we propose the concept of an \emph{infinitely divisible tradeoff function} and show that a tradeoff function has a log-concave CND if and only if it is infinitely divisible. We also give a construction to produce the unique log-concave CND for an infinitely divisible family of tradeoff functions. In Section \ref{s:piece}, we determine when a piece-wise linear tradeoff function is divisible, and show that $f_{\ep,0}$ and related tradeoff functions are not infinitely divisible, and hence do not have log-concave CNDs.

\subsection{CNDs for \texorpdfstring{$(\ep,0)$-DP}{pure-DP}}\label{s:pure}
In \citet{awan2021canonical}, it was shown that in general, the CND is not unique, but it was not clear whether there existed alternative CNDs for $f_{\ep,0}$ or $f_{\ep,\de}$. We begin this section by showing that the Tulap distribution, which was shown to be a CND for $f_{\ep,\de}$ by \citet{awan2021canonical} is in fact the \emph{unique} CND for $f_{\ep,0}$. The Tulap distribution was proposed by \citet{awan2018differentially} for the purpose of designing uniformly most powerful hypothesis tests for Bernoulli data. In the case of $(\ep,0)$-DP, the Tulap distribution coincides with one of the staircase mechanisms \citep{geng2015optimal}. It is also closely related to the discrete Laplace distribution (also known as the geometric mechanism), which is optimal for a wide range of utility functions in \citet{ghosh2012universally}.


\begin{restatable}{prop}{propuniqueCND}\label{prop:uniqueCND}
Let $\ep>0$. The distribution $\mathrm{Tulap}(0,\exp(-\ep),0)$ is the unique CND for $f_{\ep,0}$. 
\end{restatable}
\begin{proof}[Proof Sketch.]
By Lemma \ref{lem:recurrence}, the only choice in a CND is on $[-1/2,1/2]$. If the density is non-constant on $[-1/2,1/2]$, we show that the likelihood ratio is not bounded by $e^{\ep}$, violating $\ep$-DP. 
\end{proof}

Proposition \ref{prop:uniqueCND} is a surprising result in that one may expect a more natural CND than the Tulap distribution, which has a discontinuous density. However, we now know that there are no other CNDs for $(\ep,0)$-DP. In particular, there is no log-concave CND, which is the topic of the next subsection.

\subsection{Infinite divisibility and log-concavity}\label{s:logconcave}
It has been shown in \citet{dong2022gaussian} and \citet{dong2021central} that tradeoff functions built from location family log-concave distributions have very nice properties for $f$-DP. Log-concave distributions are have a monotone likelihood ratio property which gives a simple closed form expression for the tradeoff function in terms of the cdf of the log-concave distribution. 
It is easily observed that a tradeoff function with a log-concave CND satsifies a property that we call \emph{infinite divisibility}. We prove that in fact a tradeoff function has a log-concave CND if and only if it is infinitely divisble. Our proof also results in a construction to produce the unique log-concave CND. 

A continuous random variable $X$ is \emph{log-concave} if its density can be written as $g_X(x)\propto \exp(C(x))$, where $C$ is a concave function. We call a (symmetric) tradeoff function $f$ \emph{log-concave} if there exists a log-concave CND $N$ for $f$. Recall that if $N\sim F$ is a CND for $f$, then 
$f(\alpha) = F(F^{-1}(\alpha)-1)$. If $N$ is also log-concave, then $f_t(\alpha)\defeq F(F^{-1}(\alpha)-t)$ is a tradeoff function for every $t\in [0,\infty)$, and the family $\{f_t\mid t\in [0,\infty)\}$ is a monoid satisfying the assumptions of Definition \ref{def:divisible}. 

\begin{defn}\label{def:divisible}
 A tradeoff function $f$ is \emph{infinitely divisible} if there exists a monoid, under the operation of functional composition, $\{f_t\in \mscr F\mid t\geq 0\}$ containing $f$ such that 
 \begin{enumerate}
     \item $f_t\circ f_s=f_{t+s}$ for all $s,t\geq 0$,
     \item $f_s$ is nontrivial for all $s>0$, and 
     \item $f_s\rightarrow f_0=\mathrm{Id}$ as $s\downarrow 0$.
 \end{enumerate}
\end{defn}

The discussion above established that log-concave CNDs are infinitely divisible. The key result of this section is that a tradeoff function is log-concave if and only if it is infinitely divisible. We saw that it is easy to construct the infinitely divisible family given a log-concave CND. Surprisingly, we give a construction to derive the log-concave CND from the infinitely divisible family as well. This result shows an intimate relationship between properties of a tradeoff function and the possible CNDs for that tradeoff function. We will see in Section \ref{s:construction} that the property of infinite divisibility shows up again in the construction of multivariate CNDs.

\begin{restatable}{thm}{thmCNDlimit}\label{thm:CNDlimit}
A nontrivial tradeoff function $f\in \mscr F$ is log-concave if and only if it is infinitely divisible.  In particular, 
\begin{enumerate}
\item If $f$ is log-concave with log-concave CND $N\sim F$, then $\{f_t\mid t\geq 0\}$ defined by $f_t = F(F^{-1}(\alpha)-t)$ satisfies the assumptions of Definition \ref{def:divisible}.  


\item Let $f$ be infinitely divisible, with monoid $\{f_t\in \mscr F\mid t\geq 0\}$, as defined in Definition \ref{def:divisible}, such that $f=f_1$. Let $F_s$ be any CND for $f_s$ (such as constructed in Proposition \ref{prop:CNDsynthetic}). Then the following limit exists $F^*(t) \defeq \lim_{s\rightarrow 0} F_s(\frac{1}{s} t)$ and $N\sim F^*$ is the unique log-concave CND for $f$. 
Furthermore, $F^*(st)$ is the unique log-concave CND for $f_s$, for all $s>0$.
\end{enumerate}
\end{restatable}
\begin{proof}[Proof Sketch.]
It is easy to verify property 1. For property 2, we consider a subsequence $s_n=1/n!$ and observe that $F_{1/n!}(n! t)$ is a CND for $f$ at every $n$, but that as $n$ increases, the number of points at which the CND is uniquely determined also increases, by Lemma \ref{lem:recurrence}. In the limit, this sequence converges to a unique cdf, which we show has the properties of a log-concave CND.
\end{proof}

\begin{example}
We will illustrate the limit of Theorem \ref{thm:CNDlimit} on $G_1$. Let  $F_{G_{2^{-n}}}$ be the constructed cdf from Proposition \ref{prop:CNDsynthetic} for $n=0,1,2,3$. The density functions corresponding to $F_{G_{2^{-n}}}(2^{n}t)$ are plotted in Figure \ref{fig:CNDlimit}. We see that as $n$ increases, the pdfs approach that of a standard normal, which we know is the log-concave CND for $f=G_1$. 


When the construction of Theorem \ref{thm:CNDlimit} is applied to $f_{\ep,0}$, the cdf $F^*$ converges to a Laplace cdf. This seems to reflect the fact that under the limit of group privacy, $(\ep,0)$-DP converges to Laplace-DP \citet[Proposition 2.15]{dong2022gaussian}.
\begin{figure}
    \centering
    \includegraphics[width=.24\linewidth]{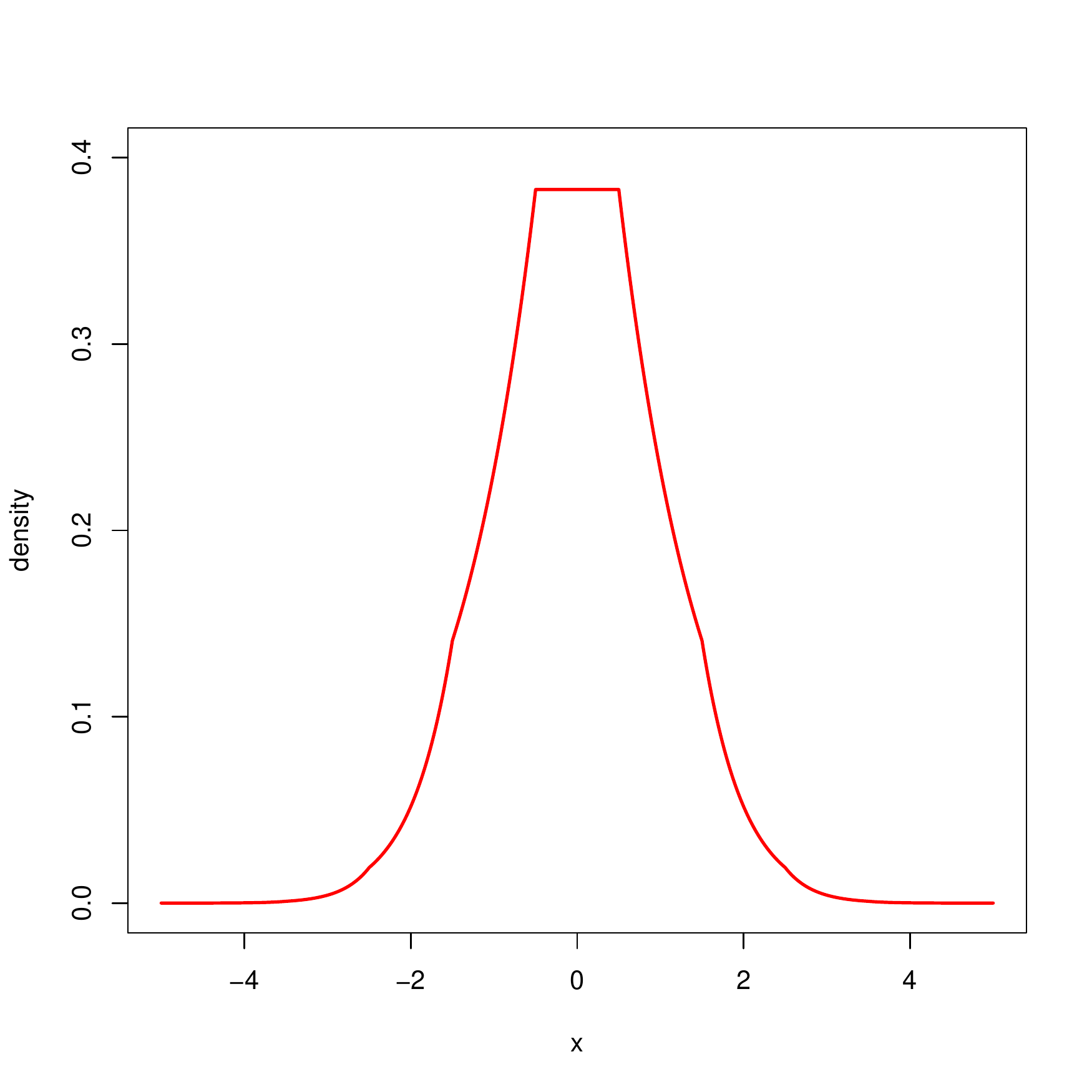}
    \includegraphics[width=.24\linewidth]{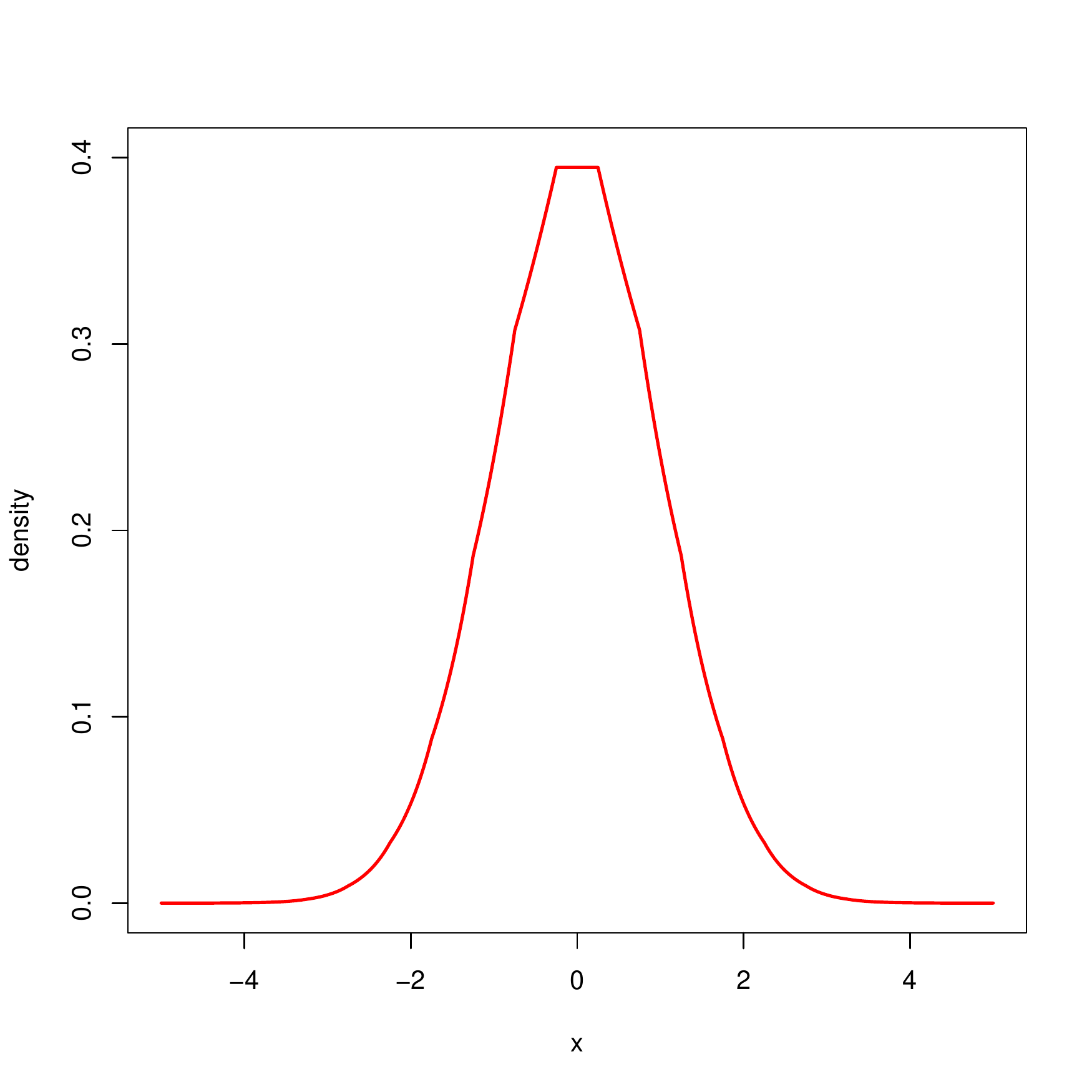}
       \includegraphics[width=.24\linewidth]{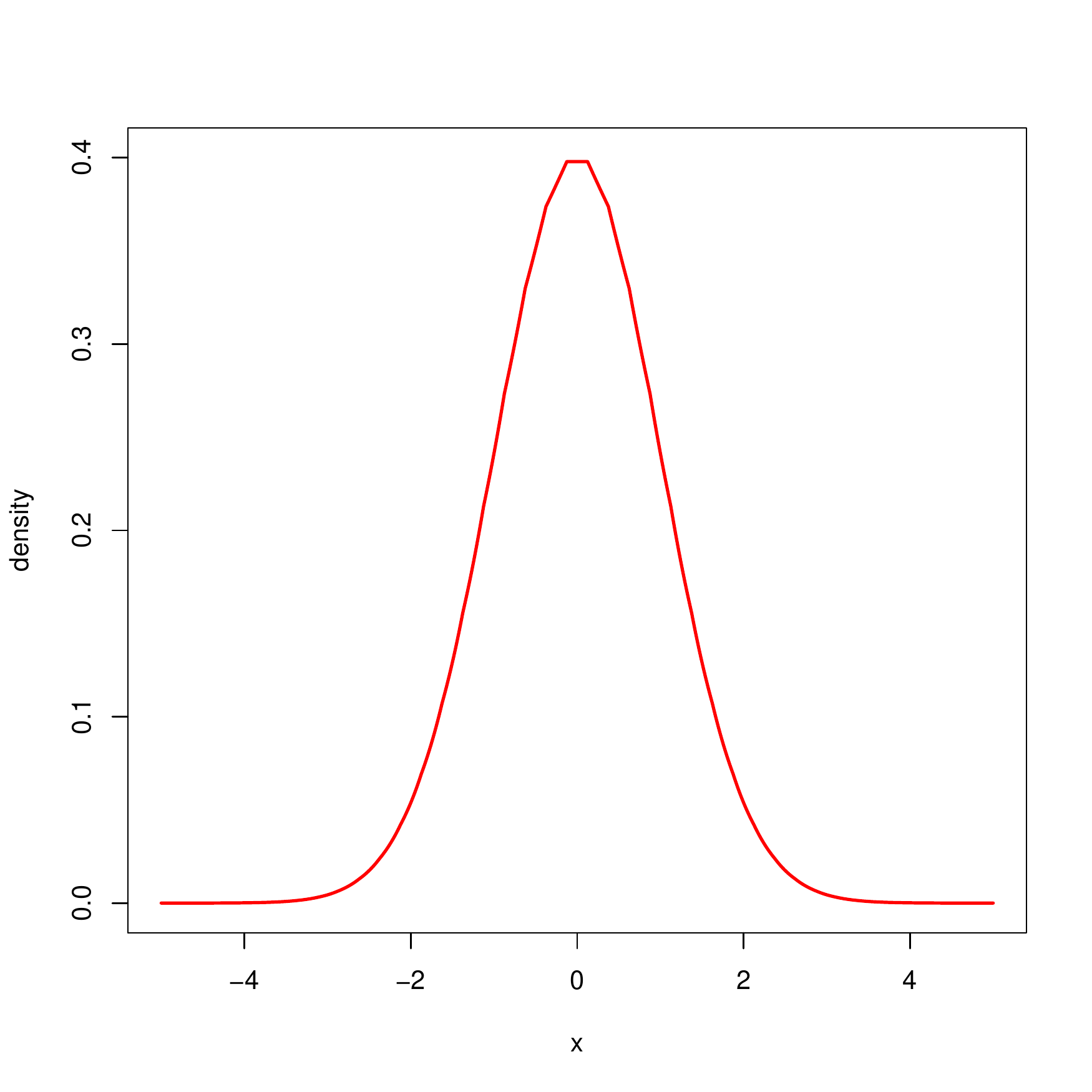}
    \includegraphics[width=.24\linewidth]{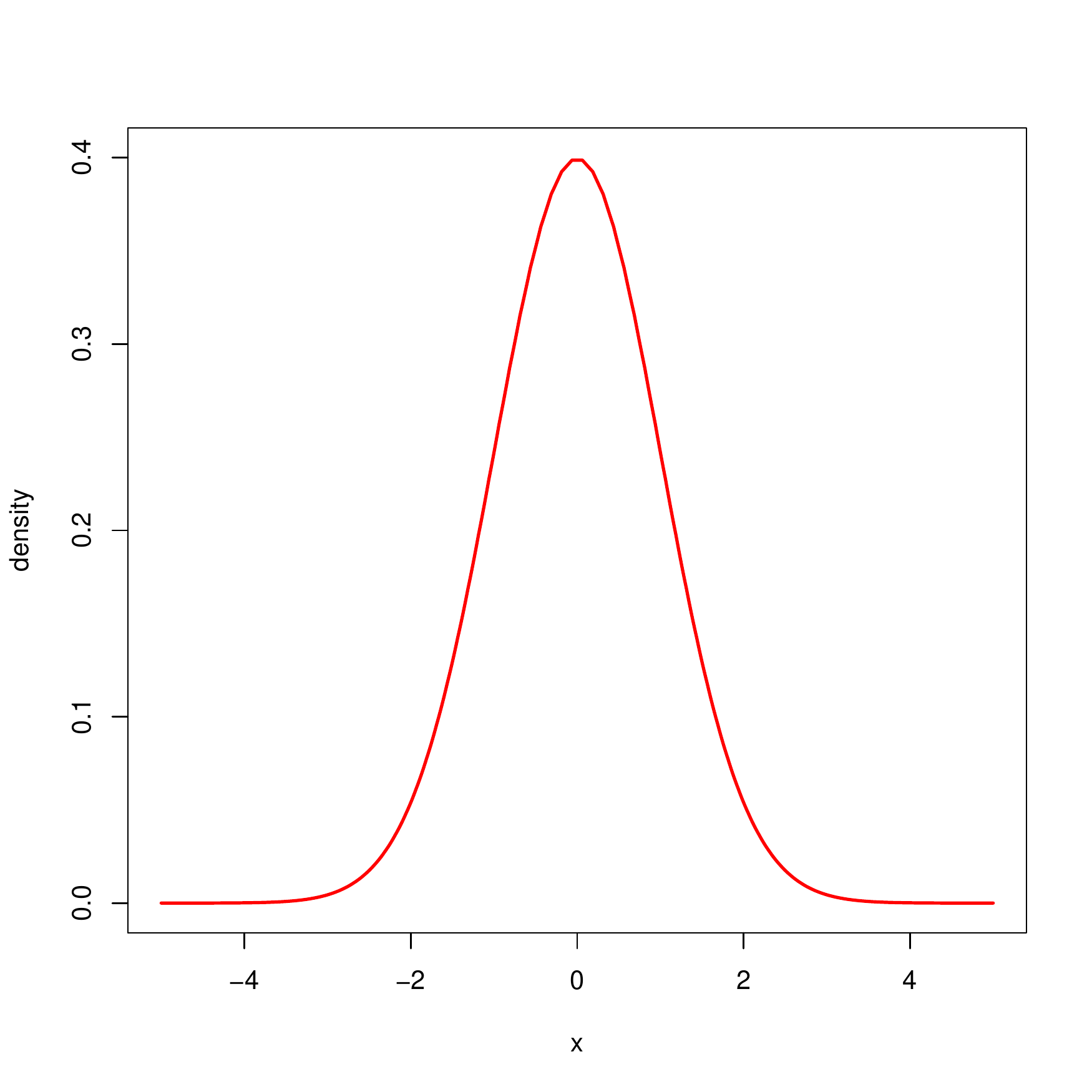}
    \caption{An illustration of Theorem \ref{thm:CNDlimit}  when applied to $G_1=T(N(0,1),N(1,1))$. From left to right, we have the density corresponding to $F_{G_{2^{-n}}}(2^n t)$ for $n=0,1,2,3$.}  
    \label{fig:CNDlimit}
\end{figure}
\end{example}

Finally, we illustrate why properties 2 and 3 of Definition \ref{def:divisible} are necessary for Theorem \ref{thm:CNDlimit}
\begin{example}
[Non examples for Theorem \ref{thm:CNDlimit}] 
 First consider why it is necessary to have $f_s\rightarrow \mathrm{Id}$. Set $f_s(\alpha) = I(\alpha=1)$ 
 for all $s>0$. Note that $f_s\circ f_t=f_{s+t}$, but that the construction of Theorem \ref{thm:CNDlimit} results in a point mass at zero, which is not a CND as it is not continuous.

 Next, suppose that all of the tradeoff functions are trivial, then $f_s(\alpha)=\alpha$ for all $s>0$, and $f_s\circ f_t=f_{s+t}$. However, there are no CNDs in this case. 
\end{example}

\subsection{Piece-wise linear tradeoff functions are generally not infinitely divisible}\label{s:piece}
We showed in Theorem \ref{thm:CNDlimit} that if a tradeoff function is infinitely divisible, then we can construct a log-concave CND. However, it is not always obvious whether a tradeoff function is infinitely divisible or not. We show that in the case of piece-wise linear tradeoff functions, we can upper bound the number of possible divisions in terms of the number of break points. In particular, the piece-wise linear tradeoff functions considered in this section are not infinitely divisible. 

We can characterize the piece-wise linear convex functions in terms of the 2nd derivative behavior: A convex function is piece-wise linear if and only if its 2nd derivative is defined everywhere except for finitely many points, and is zero whenever it is defined.

Part 1 of Proposition \ref{prop:piecewise} shows that a piece-wise linear tradeoff function $f$, which satisfies $f(x)=0$ implies $x=0$, can be sub-divided only a finite number of times. A consequence of this is that $f_{\ep,0}$ and several related tradeoff functions are not infinitely divisible and hence do not have log-concave CNDs. In fact, not only is $f_{\ep,0}$ not infinitely divisible, but there is in fact \emph{no division} $f_{\ep,0} = f\circ g$ into symmetric tradeoff functions, except where either $f$ or $g$ is the identity!

\begin{restatable}{prop}{proppiecewise}\label{prop:piecewise}
\begin{enumerate}
\item Let $f$ be a nontrivial piece-wise linear tradeoff function with $k\geq 1$ breakpoints and such that $f(x)=0$ implies that $x=0$. Then there is no tradeoff function $g$ such that $g^{\circ (k+1)}=f$.
\item Let $\ep>0$. There does not exist nontrivial symmetric tradeoff functions $f_1$ and $f_2$ such that $f_{\ep,0}=f_1\circ f_2$. 
\item Let $f$ be the tradeoff function obtained by an arbitrary sequence of mechanism compositions, functional compositions, or subsampling (without replacement) of $f_{\ep,0}$ (could be different $\ep$ values for each). Then $f$ is not infinitely divisible and so does not have a log-concave CND. 
\end{enumerate}
\end{restatable}
\begin{proof}[Proof Sketch.]
We show in Lemma \ref{lem:piecewise} that divisions of a piece-wise linear tradeoff function are themselves piece-wise linear, and that the functional composition of piece-wise linear tradeoff functions increases the number of breakpoints. This then limits the number of divisions a piece-wise linear tradeoff function can have in terms of the number of its breakpoints. 
\end{proof}


\begin{example}[$f_{0,\de}$ is log-concave]
What if $f(x)=0$ does not imply that $x=0$? The tradeoff functions $f_{0,\de}$ fit within this setting, and the results of Proposition \ref{prop:piecewise} do \emph{not} apply here. In fact, $f_{0,\de}$ is infinitely divisible with log-concave CND $U(-1/(2\de),1/(2\de))$. That is $f_{0,\de} = T(U,U+\de)$ where $U\sim U(-1/2,1/2)$. 
While $f_{\ep,\de}$ for $\de>0$ also does not satisfy the assumption that $f_{\ep,\de}(x)$ implies $x=0$, it is not clear at this time whether $f_{\ep,\de}$ is log-concave or not. 
\end{example}

\section{Multivariate CNDs}\label{s:multi}
In this section, we generalize the definition of CND to dimensions greater than one. While in the univariate case, \emph{sensitivity} is measured using the absolute distance between two statistic values, in $\RR^d$, there are many choices of norms which can be used to measure the sensitivity \citep{awan2020structure}. So, we will specify the sensitivity norm when talking about a multivariate CND. In Definition \ref{def:CND_MVT} we define a multivariate CND to be a natural generalization of properties 1-4 of Definition \ref{def:CND}. 

\begin{defn}\label{def:CND_MVT}
 Let $f$ be a symmetric tradeoff function, and let $\lVert\cdot \rVert$ be a norm on $\RR^d$. A continuous random vector $N$ with density $g$ is a \emph{canonical noise distribution} (CND) for $f$, with respect to $\lVert\cdot \rVert$, if
 \begin{enumerate}
     \item 
     For all $v\in \RR^d$ such that $\lVert v\rVert\leq 1$ we have that $T(N,N+v)\geq f$,
     \item there exists $\lVert v^*\rVert\leq 1$ such that $T(N,N+v^*)(\alpha)=f(\alpha)$ for all $\alpha \in (0,1)$,
     \item for all $v^*$ which satisfy property 2, and all $w\in \RR^d$, we have that the likelihood ratio $g(w+tv^*-v^*)/g(w+tv^*)$ is a non-decreasing function of $t\in \RR$, 
     \item $N$ is symmetric about zero: $g(x)=g(-x)$ for all $x\in \RR^d$.
 \end{enumerate}
\end{defn}

When restricted to $d=1$, Definition \ref{def:CND_MVT} recovers Definition \ref{def:CND}. This is clear for properties 1, 2, and 4. Property 3 of Definition \ref{def:CND} can be interpreted as requiring that an optimal rejection set for $T(N,N+1)$ is of the $[x,\infty)$ for some $x$. By the Neyman Pearson Lemma, we know that this holds if and only if the likelihood ratio $F'(x-1)/F'(x)$ is non-decreasing in $x$. We see that when $d=1$, property 3 of Definition \ref{def:CND_MVT} is equivalent to property 3 of Definition \ref{def:CND}. We can interpret Property 3 of Definition \ref{def:CND_MVT} as enforcing a monotone likelihood ratio in directions parallel to $v^*$. 


\subsection{Constructions of multivariate CNDs}\label{s:construction}

Composition gives a simple method to construct a multivariate CND whenever a tradeoff function can be decomposed into the composition of $k$ tradeoff functions:

\begin{restatable}{prop}{propCNDcomposition}\label{prop:CNDcomposition}
Suppose that $f = f_1\otimes f_2\otimes\cdots\otimes f_k$ all be nontrivial and symmetric tradeoff functions, and let $F_1,F_2,\ldots, F_k$ be  CNDs for $f_1,\ldots, f_k$ respectively. Let $N=(N_1,\ldots, N_k)$ be the random vector where $N_i \sim F_i$ are independent. 
Then $N$ is a CND for $f$ with respect to $\lVert \cdot \rVert_\infty$. 
\end{restatable}


Interestingly, when a tradeoff function is infinitely divisible and hence has a log-concave CND by Theorem \ref{thm:limitTradeoff}, we can create a multivariate CND with respect to $\lVert \cdot \rVert_1$-sensitivity.

\begin{restatable}{thm}{thmlonemech}\label{thm:lonemech}
Let $f$ be a nontrivial and symmetric log-concave tradeoff function with log-concave CND $F$. Let $N=(N_1,\ldots, N_k)$ be the random vector where $N_i\sim F$ are independent. 
Then $N$ is a (log-concave) CND for $f$ with respect to $\lVert \cdot \rVert_1$.
\end{restatable}
\begin{proof}[Proof Sketch.]
 Since the noise added is i.i.d., we can rephrase the tradeoff function as the tensor product of the individual tradeoff functions. We apply Theorem \ref{lem:otimesCirc} which lower bounds the tensor product of tradeoff functions with the functional composition. 
\end{proof}

Theorem \ref{thm:lonemech} was inspired by the i.i.d. Laplace mechanism. In Section \ref{s:laplace}, we show that the i.i.d. Laplace mechanism is a special case of Theorem \ref{thm:lonemech} and gives a multivariate CND for Laplace-DP. 

Note that Theorem \ref{thm:lonemech} results in a log-concave multivariate CND, and if each of $N_1,\ldots, N_k$ are log-concave in Proposition \ref{prop:CNDcomposition}, then that constructed multivariate CND is log-concave as well. \cite{dong2021central} showed that log-concave distributions have many nice properties in multivariate settings as well. We leave it to future work to investigate when multivariate log-concave CNDs exist. 

\subsection{Multivariate CND for GDP}\label{s:gdp}
Recall that if $N\sim N(0,I)$ is a $d$-dimensional Gaussian random vector, and $v\in \RR^d$ is any vector, then $T(N,N+v)=T(N(0,1),N(\lVert v\rVert_2,1))$ \citep[Proposition D.1(5)]{dong2022gaussian}. This previous result implies that $N(0,I)$ was a multivariate CND for GDP under $\lVert \cdot \rVert_2$-sensitivity. In fact, we show in Proposition \ref{prop:Gauss} that for GDP, any multivariate Gaussian is a CND with respect to any norm.

\begin{restatable}{prop}{propGauss}\label{prop:Gauss}
Let $\Sigma$ be a $d\times d$ positive definite matrix. Let $v^*\in \argmax_{\lVert u\rVert\leq 1} \lVert \Sigma^{-1/2} u\rVert_2$. Then $N(0,\Sigma)$ is a $d$-dimensional CND for $\lVert \Sigma^{-1/2}v^*\rVert_2$-GDP with respect to the norm $\lVert\cdot \rVert$.
\end{restatable}

\begin{remark}
While a multivariate Gaussian is always a multivariate CND for GDP, there is still possibly room for improvement. For Definition \ref{def:CND_MVT}, we only need a single vector to satisfy property 2. However, we could potentially ask that the bound is achieved at all $u$ such that $\lVert u \rVert=1$. Note that if $\lVert \cdot \rVert$ is an elliptical norm, then we do get this stronger property for the multivariate Gaussian, when we choose $\Sigma$ to align with the sensitivity norm. 
\end{remark}

\subsection{Multivariate CND for \texorpdfstring{$(0,\de)$-DP}{zero-delta-DP}}
First let's review a few facts about $(0,\de)$-DP, also known as $f_{0,\de}$-DP. First, note that $U(\frac{-1}{2\de},\frac{1}{2\de})$ is a (log-concave) CND for $f_{0,\de}$. So, we can write $f_{0,\de} = T(U,U+\de)$ where $U\sim U(-1/2,1/2)$. Because of this, we have that $f_{0,\de}$ is infinitely divisible, and $f_{0,\de_1}\circ f_{0,\de_2}=f_{0,\min\{\de_1+\de_2,1\}}$. Furthermore, $f_{0,\de_1}\otimes f_{0,\de_2} = f_{0,1-(1-\de_1)(1-\de_2)}$, as observed in \citet{dong2022gaussian}. This means that $f_{0,\de}$ is also infinitely decomposable, a property that we had only seen for GDP before. This decomposability implies, by Proposition \ref{prop:CNDcomposition} that we can build a multivariate CND for $f_{0,\de}$ under $\lVert \cdot\rVert_\infty$-sensitivity. In fact, this construction is a multivariate CND for any sensitivity norm.

\begin{restatable}{prop}{propCNDTVDP}\label{prop:CND_TV_DP}
Let $0<\de\leq 1$, $d\geq 1$, and $\lVert \cdot \rVert$ be a norm on $\RR^d$. Call $v^* \in \underset{\lVert v \rVert\leq 1}\arg\min\prod_{i=1}^d (1-\de|v_i|)$ and $A = \prod_{i=1}^d(1-\de|v_i^*|)$. Then $U(\frac{-1}{2\de},\frac{1}{2\de})^n$ is a CND for $f_{0,1-A}$ under $\lVert \cdot\rVert$-sensitivity. In the special case of $\lVert \cdot \rVert=\lVert\cdot \rVert_\infty$, this simplifies to $A=(1-\de)^d$. 
\end{restatable}


\subsection{Multivariate CND for \texorpdfstring{$f_{\ep,\de}$}{approximate DP} when \texorpdfstring{$\de>0$}{delta is positive}}

Let $\ep>0$ and $\de\in (0,1]$. Recall that $f_{\ep,\de}=f_{\ep,0}\otimes f_{0,\de}$  \citep{dong2022gaussian}. Since $f_{0,\de}$ is infinitely decomposable, we can write $f_{\ep,\de} = f_{\ep,0}\otimes f_{0,\de_1}\otimes \cdots\otimes f_{0,\de_k}$ where $\de = \prod_{i=1}^k (1-\de_i)$. By Proposition \ref{prop:CNDcomposition} we construct a multivariate CND for $f_{\ep,\de}$ with respect to $\lVert \cdot \rVert_\infty$-sensitivity by using $\mathrm{Tulap}(0,\exp(-\ep),0)$ in one coordinate, and the uniform distributions  $U(\frac{-1}{2\delta_i},\frac{1}{2\delta_i})$ in the other $k$ coordinates.

\subsection{Two multivariate CNDs for Laplace-DP}\label{s:laplace}
Many mechanisms designed to satisfy $(\ep,0)$-DP actually satisfy the stronger privacy guarantee of Laplace-DP. In particular, variations on the Laplace mechanism are very common additive mechanisms used to achieve $(\ep,0)$-DP. In this section, we show that two multivariate versions of the Laplace mechanism, the $\ell_1$ and $\ell_\infty$ mechanisms, are multivariate CNDs for Laplace-DP.

The Laplace distribution, denoted $\mathrm{Laplace}(m,s)$ is a distribution on $\RR$ with density $\frac{1}{2s}\exp(\frac{-1}{s} |x-m|)$. We say a mechanism satisfies $\ep$-Laplace-DP if it satisfies $L_\ep$-DP, where $L_\ep\defeq T(N,N+\ep)$ and $N\sim \mathrm{Laplace}(0,1)$. It is easily seen that $\mathrm{Laplace}(0,1/\ep)$ is a log-concave  CND for $\ep$-Laplace-DP.


{\bf i.i.d. Laplace Mechanism}
The i.i.d. Laplace mechanism is defined as follows: Let $\ep>0$ be given. If $T:\mscr X\rightarrow \RR^k$ has $\lVert\cdot\rVert_1$-sensitivity of $\Delta$, then the i.i.d. Laplace mechanism releases $T(X)+\Delta N$, where $N=(N_1,\ldots, N_k)$ is the random vector with i.i.d. entries $N_i\sim \mathrm{Laplace}(0,1/\ep)$. It is well known that the i.i.d. Laplace mechanism satisfies $f_{\ep,0}$-DP \citep[Theorem 3.6]{dwork2014algorithmic}. Since $N_1$ is a log-concave CND for $L_\ep$, Theorem \ref{thm:lonemech} shows that $N$ is a CND for $L_\ep$, with respect to $\lVert \cdot \rVert_1$-sensitivity. As $L_\ep\geq f_{\ep,0}$ and $L_\ep(\alpha)>f_{\ep,0}(\alpha)$ for some values of $\alpha$, we can more precisely capture the privacy cost of the i.i.d. Laplace mechanism using tradeoff functions rather than $\ep$-DP.

{\bf $\ell_\infty$-Mechanism} The $\ell_\infty$-mechanism, proposed in \citet{steinke2016between} is a special case of the $K$-norm mechanisms \citep{hardt2010geometry}, with density proportional to $\exp(-\ep \lVert x\rVert_\infty)$. \citet{steinke2016between} showed that the $\ell_\infty$ mechanism can improve the sample complexity of answering multiple queries, when accuracy is measured by $\ell_\infty$-norm. \citet{awan2020structure} showed that the $\ell_\infty$ mechanism is near optimal in certain applications of private linear and logistic regression. It is well known that when using $\ell_\infty$-sensitivity, the $\ell_\infty$-mechanism satisfies $\ep$-DP. In this section, we show that the $\ell_\infty$-mech is a CND for $L_\ep$, with respect to  $\ell_\infty$-sensitivity.

\begin{restatable}{prop}{propLinfty}\label{prop:Linfty}
Let $\ep>0$, and $d\geq 1$. Let $X$ be a $d$-dimensional random vector with density $g(x)=\frac{\exp(-\ep \lVert x\rVert_\infty)}{d!(2/\ep)^d}$. Then $X$ is a CND for the tradeoff function $L_\ep$ with respect to $\lVert \cdot \rVert_\infty$.
\end{restatable}
\begin{proof}[Proof Sketch.]
 First we show that with the shift of $v^*=(1,1,\ldots, 1)$, the privacy loss random variable coincides with that of $L_\ep$. Then, we show that $v^*=(1,1,\ldots, 1)^\top$ is the worst case of any shift $v$ to minimize the tradeoff functions. To deal with the case that some of the entries of $v$ are zero, we establish a convergence theorem for tradeoff functions in Theorem \ref{thm:limitTradeoff} of the Appendix.
\end{proof}

\subsection{No multivariate CND for \texorpdfstring{$f_{\ep,0}$}{pure-DP}}\label{s:nopure}
By the construction of Proposition \ref{prop:CNDsynthetic}, we know that a one-dimensional CND exists for any nontrivial tradeoff function. It turns out that the same cannot be said for the multivariate setting. In Theorem \ref{thm:noCNDpure}, we show that there is \emph{no multivariate CND for $f_{\ep,0}$ with respect to any norm}. In fact, we prove the stronger result that it is not even possible to satisfy properties 1 and 2 of Definition \ref{def:CND_MVT}

\begin{restatable}{thm}{thmnoCNDpure}\label{thm:noCNDpure}
Let $d\geq 2$ and let $\lVert\cdot \rVert$ be any norm on $\RR^d$. Then for any $\ep>0$, there is no random vector satisfying properties 1 and 2 of Definition \ref{def:CND_MVT} for $f_{\ep,0}$ with respect to the norm $\lVert \cdot \rVert$. In particular, there is no multivariate CND for $f_{\ep,0}$. 
\end{restatable}
\begin{proof}[Proof Sketch.]
Suppose to the contrary, then $(\ep,0)$-DP imposes strict bounds on the likelihood ratio of the distribution. These bounds allow us to find an arbitrarily long sequence of points, sufficiently far apart, where the density is bounded below. This ultimately shows that the density is not integrable.
\end{proof}
Combining Theorem \ref{thm:noCNDpure} with Proposition \ref{prop:CNDcomposition}, we infer in Corollary \ref{cor:pureDecomp} that $f_{\ep,0}$ cannot be written as the tensor product of any two nontrivial tradeoff functions. This means that if we want to design two independent mechanisms such that the joint release exactly satisfies $(\ep,0)$-DP, then one of the mechanisms must be perfectly private. 

\begin{restatable}{cor}{corpureDecomp}\label{cor:pureDecomp}
Let $\ep>0$ be given. There does not exist nontrivial symmetric tradeoff functions $f_1$ and $f_2$ such that $f_{\ep,0}=f_1\otimes f_{2}$.
\end{restatable}

\begin{remark}
Theorem \ref{thm:noCNDpure} along with Theorem  \ref{thm:lonemech} gives an alternative argument that $f_{\ep,0}$ is not log-concave/infinitely decomposable. 
\end{remark}

\section{Discussion}\label{s:discussion}
Motivated by the goals of constructing log-concave CNDs and multivariate CNDs, we found some fundamental connections between these constructions and the operations of mechanism composition and functional composition of the tradeoff functions. Surprisingly, the constructions for both log-concave and multivariate CNDs relied on whether a tradeoff function could be decomposed either according to functional composition, or according to mechanism composition. An interesting result of our work was that for $(\ep,0)$-DP there is a unique 1-dimensional CND and no multidimensional CNDs, which implies that $f_{\ep,0}$ can neither be decomposed according to functional composition or mechanism composition. This highlights the limitations of pure-DP as a privacy definition. On the other hand, Gaussian-DP, Laplace-DP, and $(0,\de)$-DP were seen to have much better properties. 

  While the framework of GDP and related notions (e.g., zero-concentrated DP) have many desirable properties, including those developed in this paper, there are still many reasons why one may be interested in other DP frameworks. In some applications, having a stronger notion of DP is needed to protect events with small probability, such as pure-DP or Laplace-DP; in this case, our work shows that Laplace-DP is a much better behaved notion of privacy than pure-DP. One may also propose other alternative DP definitions based on a family of tradeoff functions, and our research gives some fundamental insights on what properties that family must have in order for log-concave or multivariate CNDs to be constructed.

We showed that a multivariate extension of CND can capture the same properties as in the 1-dimensional case. \citet{awan2021canonical} showed that in one dimension, CNDs can be used to obtain DP hypothesis tests with optimal properties. An open question is whether our definition of a multivariate CND has any connections to optimal hypothesis testing. 

Most of the constructions of multivariate CNDs presented in this paper are product distributions. Even the multivariate CNDs for GDP are a linear transformation of i.i.d. random variables. The $\ell_\infty$-mechanism is the exception, providing a truly nontrivial CND for Laplace-DP. It is worth exploring whether there are general techniques to produce nontrivial multivariate CNDs like the $\ell_\infty$-mechansism, as well as exploring the merits of such CNDs.

While many of the multivariate CNDs constructed for tradeoff functions, only held for specific sensitivity norms, a more general question would be on the existence and construction of multivariate CNDs for an arbitrary tradeoff/norm pair.

\section*{Acknowledgments}

This work was supported in part by NSF SES 2150615, awarded to Purdue University.


\bibliographystyle{plainnat} 
\bibliography{bibliography} 

\appendix
\section{Appendix}
\subsection{Broader impacts}\label{s:broader}
Privacy is an important societal problem, and there is a natural tradeoff between the privacy afforded to the individuals of the dataset, and the utility of the published result. One may be concerned that differential privacy techniques reduce the utility of the results too much, in exchange for the privacy protection. In our work, by providing a better understanding of differential privacy, and by developing new mechanisms to achieve differential privacy, we make it possible to achieve higher utility at the same privacy cost; or alternatively, we can maintain the same utility while giving a stronger privacy protection. In our view, optimizing the privacy-utility tradeoff is universally beneficial to society, and we do not foresee any negative societal impacts of this work.

\subsection{Relations between functional composition and tensor product}
Both the functional composition and the tensor product of tradeoff functions are essential concepts in our constructions of CNDs. In the remainder of this section, we establish some new relations between the two operations, which provide an interesting insight into the connection between group privacy and composition. First, we recall a lemma from \citet{dong2022gaussian}:




\begin{lemma}[Lemma A.5: \citet{dong2022gaussian}]\label{lem:A5}
Suppose that $T(P,Q)\geq f$ and $T(Q,R)\geq g$. Then $T(P,R)\geq g\circ f$. 
\end{lemma}

\begin{restatable}{lem}{lemotimesCirc}\label{lem:otimesCirc}
Let $f$ and $g$ be any two symmetric tradeoff functions. Then $f\otimes g\geq f \circ g$. 
\end{restatable}
\begin{proof}
    First note that if either $f$ or $g$ is equal to $\mathrm{Id}$, then the result is trivial. Now, suppose that both $f$ and $g$ are nontrivial, and let $N_1\sim F$ and $N_2\sim G$ be independent, where $F$ is a CND for $f$ and $G$ is a CND for $g$. 
    
    By definition of the tensor product of tradeoff functions \citep[Definition 3.1]{dong2022gaussian}, we have that 
    \begin{equation}T\left( \binom{0+N_1}{0+N_2},\binom{1+N_1}{1+N_2}\right)=f\otimes g,
    \label{eq:otimes}
    \end{equation}
    since $T(0+N_1,1+N_1)=f$ and $T(0+N_2,1+N_2)=g$, by definition of CND.
    
    It is also true that 
    \begin{align*}
        T\left(\binom{0+N_1}{0+N_2},\binom{1+N_1}{0+N_2}\right)&=f,\\
        T\left(\binom{1+N_1}{0+N_2},\binom{1+N_1}{1+N_2}\right)&=g.
    \end{align*}
        Then by Lemma \ref{lem:A5}, 
        \begin{equation}
            T\left( \binom{0+N_1}{0+N_2},\binom{1+N_1}{1+N_2}\right)\geq g\circ f=f\circ g,
            \label{eq:circ}
        \end{equation}
        where the last equality follows since $f$ and $g$ are symmetric, using \citep[Lemma A.4]{dong2022gaussian}. Comparing Equations \eqref{eq:otimes} and \eqref{eq:circ}, we have that $f\otimes g\geq f\circ g$.
\end{proof}

\begin{remark}
As a special case of Lemma \ref{lem:otimesCirc}, we have that $f\otimes f\geq f\circ g$, which has an interesting interpretation: Suppose two situations: 1)  your data is present once in two databases, and an $f$-DP mechanism is applied to each database once. This gives $f\otimes f$-DP cumulative privacy cost to you. 2) your data is present twice in one database, and an $f$-DP mechanism is applied once to the database. Since your data is present twice, by group privacy the incurred privacy cost to you is $f\circ f$-DP. Lemma \ref{lem:otimesCirc} says that you would prefer to be in the two separate databases. The intuition behind this can be understood as follows: in the second scenario, the privacy expert could choose to split the database into two: each one containing a copy of your data, and apply an $f$-DP mechanism to both. The nominal privacy cost of this would be $f$-DP (considering groups of size 1), as changing one entry affects only one of the two calculations. However, for groups of size two, the privacy cost is $f\circ f$-DP. This shows that all of the mechanisms in scenario 1 could also be applied to scenario 2, but in general there are mechanisms in scenario two that are not possible in scenario 1 (since in scenario 1, the databases cannot be merged). 
\end{remark}

Before we move on, we give a Lemma, extending \citet[Lemma E.8]{awan2021canonical} to arbitrary $k$. Lemma \ref{lem:groupK} shows that given a CND $N$ for $f$, we can easily produce a CND for $f^{\circ k}$ by rescaling $N$ by $\frac 1k$. To establish Lemma \ref{lem:groupK}, we need another technical lemma, which appeared within the proof \citet[Lemma E.12]{awan2021canonical}. We say that a cdf $F$ is \emph{invertible} at $t$ if $F^{-1}(F(t))=t$. 

\begin{lemma}[\citet{awan2021canonical}]\label{lem:invertible}
Let $f$ be a nontrivial symmetric tradeoff function, and let $F$ be a CND for $f$. Call $M\defeq \inf\{t\mid 0<F(t)\}$. Then if $M< -1/2$ and $\alpha> 1-f(1)$, then $F$ is invertible at $F^{-1}(\alpha)-1$.
\end{lemma}
\begin{proof}
Let $\alpha>1-f(1)$, or equivalently $1-\alpha<f(1)$. 
Note that $F$ is invertible at $F^{-1}(\alpha)$, and by \citet[Lemma E.3]{awan2021canonical} $F$ is also invertible at $F^{-1}(\alpha)-1$ unless $F^{-1}(\alpha)-1<M$. So, we need to show that $F^{-1}(\alpha)\geq M+1$:

\begin{align}
    M+1&=\inf \{t+1\mid 0<F(t)\}\\
    &=\inf \{t\mid 0<F(t-1)\}\\
        &=\inf \{t\mid 1>1-F(t-1)\}\\
    &=\inf\{t\mid 1-f(1)<1-f(1-F(t-1))\ \&\ 0<F(t-1)\}\label{eq:twice1}\\
    &=\inf\{t\mid 1-f(1)<F(t)\ \& \ 0<F(t-1)\}\label{eq:twice2},
\end{align}
where \eqref{eq:twice1} uses the fact that $1-f$ is strictly decreasing at $1$; \eqref{eq:twice2} uses the fact that $0<F(t-1)$ to apply the recursion of Lemma \ref{lem:recurrence}. Now, suppose that $F(t-1)=0$: then $t-1\leq M$ and because $M<-1/2$, $F(t)<1$. So, $0=F(t-1)$ implies that $0=f(F(t))$. But this in turn implies that $F(t)\leq 1-f(1)$. We see that $1-f(1)<F(t)$ implies that $0<F(t-1)$. So, 

\begin{align}
   M+1 &=\inf\{t\mid 1-f(1)<F(t)\}\\
    &\leq\inf\{t\mid \alpha\leq F(t)\}\label{eq:twice4}\\
    &=F^{-1}(\alpha),
    \end{align}
    where \eqref{eq:twice4} uses the fact that $\alpha>1-f(1)$. We see that $F^{-1}(\alpha)-1\geq M$ and conclude that $F$ is invertible at $F^{-1}(\alpha)-1$. 
\end{proof}

\begin{restatable}{lem}{lemgroupK}\label{lem:groupK}
Let $F$ be a CND for a nontrivial symmetric tradeoff function $f$. Then $F(k\cdot)$ is a CND for $f^{\circ k}$ for any $k\in \NN^+$.
\end{restatable}
\begin{proof}
For any $k$, denote $F_k(x) = F(kx)$ and $F^{-1}_k(x) = \frac 1k F^{-1}(x)$, where $F^{-1}_k$ is the quantile function of $F_k$. Symmetry and continuity of $F_k$ are clear. 

 For induction, assume that for some $k>1$, we have that $F_{k-1}$ is a CND for $f^{\circ(k-1)}$. In particular, we have that 
    \begin{align*}
        f^{\circ (k-1)}&=F_{k-1}(F^{-1}_{k-1}(\alpha)-1)\\
        &=F[(k-1)\{[1/(k-1)]F^{-1}(\alpha)-1\}]\\
        &=F(F^{-1}(\alpha)-(k-1)).
    \end{align*} 
    
    Let $M\defeq \inf\{t\mid 0<F_{k-1}(t)\}$. By symmetry of $F_{k-1}$, we know that $M\leq 0$. If $M\geq -1/2$, then we have that $f^{\circ k}(\alpha)\leq f^{\circ (k-1)}(\alpha)=F_{k-1}(F_{k-1}^{-1}(\alpha)-1)=0$ for all $\alpha\in (0,1)$; we also have $F_k(F_k^{-1}(\alpha)-1)=F(F^{-1}(\alpha)-k)=0=f^{\circ k}(\alpha)$. Furthermore, $T(F_k(\cdot),F_k(\cdot-1))=T(F(\cdot),F(\cdot-k))=0=g$, since $F(\cdot)$ and $F(\cdot-k)$ have disjoint support. Finally, note that $T(F_k(\cdot),F_k(\cdot-m))\geq 0=T(F_k(\cdot),F_k(\cdot-1))$, since $0$ is a trivial lower bound for any tradeoff function. We conclude that when $M\geq -1/2$, $F_k$ is a CND for $f^{\circ k}$.
    
    Now suppose that $M<-1/2$ and let $\alpha\in (0,1)$. If $\alpha\leq 1-f^{\circ (k-1)}(1)$, then $f^{\circ k}(\alpha)=f(f^{\circ (k-1)}(\alpha))=f(0)=0$ and $F(F^{-1}(1-\alpha)-k)\leq F(F^{-1}(1-\alpha)-(k-1))=f^{\circ (k-1)}(\alpha)=0=f^{\circ k}(\alpha)$ because $F$ is increasing. We see that $f^{\circ k} = F_k(F_k^{-1}(\alpha)-1)$ in this case. Now assume that $\alpha>1-f(1)$. Then by Lemma \ref{lem:invertible}, we have that $F_{k-1}$ is invertible at $F^{-1}_{k-1}(\alpha)-1$. Then,
    
    \begin{align}
        f^{\circ k}&=f\circ F_{k-1}(F^{-1}_{k-1}(\alpha)-1)\\
        &=F(F^{-1}[F_{k-1}(F^{-1}_{k-1}(\alpha)-1)]-1)\\
        &=F((k-1)F^{-1}_{k-1}[F_{k-1}(F^{-1}_{k-1}(\alpha)-1)]-1)\\
        &=F((k-1)(F^{-1}_{k-1}(\alpha)-1)-1)\label{eq:invertible}\\
        &=F(F^{-1}(\alpha)-k)\\
        &=F_{k}(F^{-1}_k(\alpha)-1),
    \end{align}
    where in \eqref{eq:invertible}, we used the fact that $F_{k-1}$ is invertible at $F_{k-1}^{-1}(\alpha)-1$.

We have shown that $f^{\circ k} = F_k(F^{-1}_k(\alpha)-1)$. Since $F_k(F^{-1}_k(\alpha)-1)$ represents the type II error of the (potentially suboptimal) test, which rejects when the observed random variable is above a certain threshold, we have that $T(N_k,N_k+1)=T(N,N+k)\leq f^{\circ k}$, where $N\sim F$ and $N_k\sim F_k$. To verify properties 2 and 3 of Definition \ref{def:CND}, it remains to show that $T(N_k,N_k+1)\geq f^{\circ k}$. Note that $T(N,N+(k-1))=f^{\circ (k-1)}$, and $T(N+(k-1),N+k)=T(N,N+1) = f$. By Lemma \ref{lem:A5}, we have that $T(N_k,N_k+1)=T(N,N+k)\geq f^{\circ (k-1)}\circ f = f^{\circ k}$, which completes the argument for parts 2 and 3 of Definition \ref{def:CND}. 

For property 1 of Definition \ref{def:CND}, let $m\in [0,1]$. As before, we use the notation $N\sim F$, $N_k\sim F_k$ and $N_{k-1}\sim F_{k-1}$. We will show that $T(N_k,N_k+m)=T(N,N+km)\geq f^{\circ k}$. If $m\leq \frac{k-1}{k}$, then $m^*=\frac{km}{k-1}\in [0,1]$. In this case, 
\begin{align*}
    T(N,N+km)&=T\left(N,N+(k-1)\frac{km}{k-1}\right)\\
    &=T(N_{k-1},N_{k-1}+m^*)\\
    &\geq T(N_{k-1},N_{k-1}+1)\\
    &=f^{\circ (k-1)}\\
    &\geq f^{\circ k},
\end{align*}
where we used the fact that $F_{k-1}$ is a CND for $f^{\circ (k-1)}$ and that $f^{\circ (k-1)}\geq f^{\circ k}$. 
If $m\geq \frac{k-1}{k}$, then $m^* = km-(k-1)\in [0,1]$. Then
\begin{align}
    T(N,N+km)&=T(N,N+(k-1)+m^*)\\
    &\geq T(N,N+(k-1))\circ T(N,N+m^*)\label{eq:A5}\\
    &=f^{\circ(k-1)}\circ T(N,N+m^*)\\
    &\geq f^{\circ (k-1)}\circ f\label{eq:groupK2}\\
    &=f^{\circ k},
\end{align}
where for \eqref{eq:A5} we use Lemma \ref{lem:A5} and the fact that $T(N+(k-1),N+(k-1)+m^*)=T(N,N+m^*)$, and for \eqref{eq:groupK2}, we use the inductive hypothesis that $F_{k-1}$ is a CND for $f^{\circ (k-1)}$.  
\end{proof}

\begin{example}[Composition and Group Privacy do not Commute]\label{ex:doesn'tcommute}
It is an interesting question whether the following property holds: $(f\otimes g)^{\circ k}=f^{\circ k} \otimes g^{\circ k}$. This is true for GDP:
\[(G_{\mu_1}\otimes G_{\mu_2})^{\circ k}=G_{k\lVert \binom{\mu_1}{\mu_2}\rVert}=G_{\lVert \binom{k\mu_1}{k\mu_2}\rVert}=G_{\mu_1}^{\circ k} \otimes G_{\mu_2}^{\circ k}.\]
However, by studying $(0,\de)$-DP, we see that this property does not hold in general -- even for log-concave tradeoff functions. We compute that 
\[(f_{0,\de_1}\otimes f_{0,\de_2})^{\circ k}=f_{0,1-(1-\de_1)(1-\de_2)}^{\circ k}=f_{0,\min\{1,k(1-(1-\de_1)(1-\de_2))\}},\]
whereas $f_{0,\de_1}^{\circ k} \otimes f_{0,\de_2}^{\circ k}=f_{0,\min\{1,k\de_1\}}\otimes f_{0,\min\{1,k\de_2\}}=f_{0,1-(1-\min\{1,k\de_1\})(1-\min\{1,k\de_2\})}$. 
plugging in $k=2$ and $\de_1=\de_2=.1$, we get that the first expression yields .38, whereas the second gives .36. Interestingly, it seems that accounting for group privacy \emph{first}, before applying composition gives the tighter privacy analysis. This is confirmed by the inequality in Proposition \ref{prop:otimesCirc}.
\end{example}

\begin{restatable}{prop}{propotimesCirc}\label{prop:otimesCirc}
	Let $f$ and $g$ be tradeoff functions. Then $(f\otimes g)^{\circ k}\leq f^{\circ k} \otimes g^{\circ k}$.
\end{restatable}

\begin{proof}
	We know that , $f=T(X,X+1)$ and $g=T(Y,Y+1)$, where $X$ is a CND for $f$ and $Y$ is a CND for $g$.
	\begin{align*}
		(f\otimes g)^{\circ k}
		&= \big(T[X,X+1]\otimes T[Y,Y+1]\big)^{\circ k}\\
		&= \big(T[(X,Y),(X+1,Y+1)]\big)^{\circ k}\\
		& \leq T[(X,Y),(X+k,Y+k)]\qquad\text{by Lemma \ref{lem:A5} and Lemma \ref{lem:groupK}} \\
		&= T[X,X+k]\otimes T[Y,Y+k]\\
		&= T[X/k,X/k+1]\otimes T[Y/k,Y/k+1]\\
		&= f^{\circ k}\otimes g^{\circ k},
	\end{align*}
	since $X/k$ and $Y/k$ are CNDs for $f^{\circ k}$ and $g^{\circ k}$ respectively, by Lemma \ref{lem:groupK}.
\end{proof}

\subsection{A limit theorem for tradeoff functions}
Below, we introduce a limit theorem for tradeoff functions, which can be used to show a mechanism satisfies $f$-DP in terms of certain limits.
\begin{restatable}{thm}{thmlimitTradeoff}\label{thm:limitTradeoff}
Let $P_n\overset {\TV}\rightarrow P$ and $Q_n\overset {\TV} \rightarrow Q$ be two sequences of distributions, which converge in total variation. Then $T(P_n,Q_n)\rightarrow T(P,Q)$ uniformly. 
\end{restatable}
\begin{proof}
By \citet[Lemma A.7]{dong2022gaussian}, it suffices to prove point-wise convergence. First we will establish $T(P,Q)$ as an asymptotic lower bound on $T(P_n,Q_n)$. By Lemma \ref{lem:A5}, we have that 
\[T(P_n,Q_n)\geq T(Q,Q_n)\circ T(P,Q)\circ T(P_n,P).\]
Since $P_n\overset {\TV}\rightarrow P$ and $Q_n\overset {\TV}\rightarrow Q$, we have that $T(P_n,P)(\alpha)\geq [(\alpha-TV(P_n,P))]_0^1$ and $T(Q_n,Q)(\alpha) \geq [(\alpha-TV(Q_n,Q))]_0^1$, where $[x]_a^b\defeq \max\{\min\{x,b\},a\}$ is the clamping function. Since all tradeoff functions are increasing, the following inequality holds: 
\[T(P_n,Q_n)\geq [(\alpha-\TV(Q_n,Q))]_0^1\circ T(P,Q)\circ [(\alpha-\TV(P_n,P))]_0^1\rightarrow T(P,Q),\]
and the limit holds since $\TV(P_n,P)\rightarrow \mathrm{Id}$, $\TV(Q_n,Q)\rightarrow \mathrm{Id}$, and tradeoff functions are continuous.

Next, we show that $T(P,Q)$ is an asymptotic upper bound for $T(P_n,Q_n)$. It suffices to check for $\alpha\in (0,1)$, since tradeoff functions are continuous. Let $\alpha^*\in (0,1)$ be given. Let $\phi$ be an optimal test for $T(P,Q)$ such that $\EE_P\phi=\alpha^*$ and $\EE_Q\phi = 1-f(1-\alpha^*)$. Note that if $U\sim U(0,1)$, we can write 
\[\EE_P\phi= \EE_{X\sim P,U} I(U\leq \phi(X))=\EE_U P_{X\sim P}(U\leq \phi(X)|U)=\EE_U P(\phi^{-1}([U,1])|U),\]
which will allow us to apply the total variation convergence. Call $\alpha_n = \EE_{P_n}\phi$ for all $n$. Notice that $\alpha_n\rightarrow \alpha^*$, since 
\begin{align*}
    |\alpha_n-\alpha^*| &= |\EE_{P_n}\phi-\EE_P\phi| \\
&=\Big|\EE_U P_n(\phi^{-1}([U,1])\mid U) - \EE_U P(\phi^{-1}([U,1])\mid U)\Big|\\
&\leq \EE_U\Big|P_n(\phi^{-1}([U,1])\mid U) -P(\phi^{-1}([U,1])\mid U)\Big|\\
& \leq \EE_U \mathrm{TV}(P_n,P)\\
&\rightarrow 0,
\end{align*}
as $P_n\overset \TV \rightarrow P$. Similarly, we have that $|\EE_{Q_n}\phi - \EE_Q\phi|\rightarrow 0$, implying that $\EE_{Q_n}\phi \rightarrow 1-f(1-\alpha^*)$. Then,

\begin{align*}
    T(P_n,Q_n)(\alpha_n)&\leq 1-\EE_{Q_n}\phi\\
    &\rightarrow f(1-\alpha^*)\\
    &=T(P,Q)(\alpha^*). 
\end{align*}

However, we actually want to show that $T(P_n,Q_n)(\alpha^*)$ is asymptotically upper bounded by $T(P,Q)(\alpha)$. Luckily, $T(P_n,Q_n)(\alpha_n)$ and $T(P_n,Q_n)(\alpha^*)$ are close for large $n$, since tradeoff functions are ``locally Lipschitz.''  We explain as follows: Since $\alpha_n\rightarrow \alpha^*$, let $N$ be such that for all $n\geq N$, $\alpha_n \in \left(0,\alpha^* + \frac{1-\alpha^*}{2}\right)$. On the interval $(0,\alpha^* + \frac{1-\alpha^*}{2})$, we claim that $T(P_n,Q_n)$ is $\frac{2}{1-\alpha^*}$-Lipschitz. This is because the derivative (or subderivative) of a convex function is increasing, and the tangent lines of a convex function are always a lower bound. In the worst case, using the points $(\alpha^*+(1-\alpha^*)/2,0)$ and $(1,1)$, the slope at $\alpha^*+(1-\alpha^*)/2$ is at most $\frac{1-0}{1-(\alpha^*+(1-\alpha^*)/2}=\frac{2}{1-\alpha^*}$. Now that we have established that $T(P_n,Q_n)$ is $\frac{2}{1-\alpha^*}$-Lipschitz on $(0,\alpha^*-(1-\alpha^*)/2)$, we have that for all $n\geq 0$, 
\[\Big|T(P_n,Q_n)(\alpha_n)-T(P_n,Q_n)(\alpha^*)\Big|\leq \frac{2}{1-\alpha^*}|\alpha_n-\alpha|\rightarrow 0.\]
We conclude that $T(P_n,Q_n)(\alpha^*)$ is asymptotically upper bounded by $T(P,Q)(\alpha^*)$ for all $\alpha*\in (0,1)$. Combining the asymptotic lower and upper bounds, we have that $T(P_n,Q_n)\rightarrow T(P,Q)$.
\end{proof}

Two immediate corollaries of the above theorem are as follows. The first, generally states that if we establish a lower bound on $T(P_n,Q_n)$ for all $n$, and $P_n\overset {\TV}\rightarrow P$ and $Q_n\overset \TV\rightarrow Q$, then the lower bound applies to $T(P,Q)$ as well. This could be generalized to a sequence of lower bounds $f_n\rightarrow f$ as well. 

\begin{restatable}{cor}{corlimitbound}\label{cor:limitBound}
Let $P_n\overset {\TV}\rightarrow P$ and $Q_n \overset {\TV}\rightarrow Q$ be two sequences of distributions such that $T(P_n,Q_n)\geq f$ for all $n$. Then $T(P,Q)\geq f$. 
\end{restatable}

Corollary \ref{cor:limitDP} shows that the limit of an $f$-DP mechanism satisfies $f$-DP (could also replace each $f$ with $f_n\rightarrow f$). This is similar to the limit result of \citet{kifer2012private}, but is phrased in terms of convergence in total variation rather than almost sure convergence.

\begin{restatable}{cor}{corlimitDP}\label{cor:limitDP}
Let $M_n$ be a sequence of mechanisms satisfying $f$-DP (i.e., $T(M_n(D),M_n(D'))\geq f$ for all adjacent $D$ and $D'$), and suppose that $M_n(D)\overset \TV \rightarrow M(D)$ for all $D$. Then $M$ satisfies $f$-DP: $T(M(D),M(D'))\geq f$. 
\end{restatable}

\subsection{Proofs and technical lemmas for the main paper}
For any measurable set $A$, let $\lambda(A)$ denote the Lebesgue measure of $A$. 

\begin{lemma}\label{lem:LebesgueShift}
Let $A,B\subset[-1/2,1/2]$ be disjoint sets with positive Lebesgue measure such that $A\cup B=[-1/2,1/2]$. Then there exists a shift $\omega\in (-1,1)$ such that $(B+\omega)\cap A$ has positive Lebesgue measure.
\end{lemma}
\begin{proof}
Suppose to the contrary that for all $\omega\in (-1,1)$, $\lambda((B+\omega)\cap A)=0$. This implies that
\begin{align}
0&=\int_{-1}^1 \lambda((B+\omega)\cap A) \ d\omega\\
&=\int_{-1}^1 \int_{-1/2}^{1/2} I(x\in (B+\omega)\cap A)\ dx\ d\omega\\
&=\int_{-1}^1\int_A I(x\in B+\omega)\ dx\ d\omega\\
\text{(Tonelli's Theorem)}&= \int_A\int_{-1}^1 I(x\in B+\omega)\ d\omega\ dx\label{eq:tonelli}\\
&=\int_A \lambda(x-B)\ dx\label{eq:U01}\\
&= \int_A \lambda(B)\ dx\label{eq:lambdaB}\\
&=\lambda(A)\lambda(B),
\end{align}
where we used Tonelli's Theorem in \eqref{eq:tonelli} to change the order of integration, in \eqref{eq:U01} we used the fact $x-B\subset [-1,1]$ since both $x$ and $B$ lie in $[-1/2,1/2]$, and in \eqref{eq:lambdaB} we used the fact that Lebesgue measure is translation invariant. We see that either $\lambda(A)=0$ or $\lambda(B)=0$, giving a contradiction.\qedhere
\end{proof}

\propuniqueCND*
\begin{proof}
Let $g$ be the density of an arbitrary CND for $f_{\ep,0}$, and let $G$ denote its cdf function. Since $g$ is the density of a symmetric random variable centered at zero, $g(x)=g(-x)$ for all $x\in \RR$. By \citet[Proposition 3.7]{awan2021canonical}, we have that $G(x+1)=1-f_{\ep,0}(1-G(x))$, which implies that $g(x+1) = f'_{\ep,0}(1-G(x))g(x)$. For $x>0$, we have that $G(x)\geq 1/2>c$, where $c$ satisfies $f_{\ep,0}(1-c)=c$. Recall that $f_{\ep,0}(\alpha)=\alpha e^{-\ep}$ for all $\alpha\leq c$. Then $f'_{\ep,0}(1-G(x))=e^{-\ep}$ for $x\geq 0$. We see that we can write $g(u+k)=e^{-|k|\ep} g(u)$ for $u\in [-1/2,1/2]$ and $k\in \ZZ$. 

We see that so far, $g$ has the freedom to choose its values in $[-1/2,1/2]$ and then all other values are determined by the above recurrence. Note that for the Tulap distribution, its density is the constant value of $\frac{\exp(\ep)-1}{\exp(\ep)+1}$ on $(-1/2,1/2)$, since it is the constructed CND by Proposition \ref{prop:CNDsynthetic} \citep[Corollary 3.10]{awan2021canonical}. Suppose that $g(u)$ is non constant on $(-1/2,1/2)$. Then it must take on some values above and below $\frac{\exp(\ep)-1}{\exp(\ep)+1}$ in order to still integrate to 1. To rule out trivial cases, where $g(x)$ is equivalent to the Tulap density up to a set of measure zero, we assume that the sets
\[A\defeq \left\{u\in (-1/2,1/2) \mid g(u)>\frac{\exp(-\ep)-1}{\exp(\ep)+1}\right\},\]
\[B\defeq \left\{u\in (-1/2,1/2) \mid g(u)\leq\frac{\exp(-\ep)-1}{\exp(\ep)+1}\right\},\]
both have positive Lebesgue measure. We denote $\lambda$ as the Lebesgue measure.

By Lemma \ref{lem:LebesgueShift}, we know that there exists $\omega\in (-1,1)$ such that $(B+\omega)\cap A)$ has positive Lebesgue measure. By symmetry of $g$ about zero, there exists a positive shift $\omega\in (0,1)$ such that $\lambda((B+\omega)\cap A)>0$. Consider $\Delta \defeq 1-\omega \in (0,1)$. Let $v\in (B+\omega)\cap A$ and $u\defeq v-\omega \in B\cap (A-\omega)$,  and consider the likelihood ratio:
\[\frac{g(1+u-\Delta)}{g(1+u)}=\frac{g(u+\omega)}{g(1+u)}=\frac{g(v)}{e^{-\ep}g(u)}>e^\ep,\]
where we use the fact that $u\in B$, $v\in A$, and $\lambda((B+\omega)\cap A)>0$. This means that the likelihood ratio $g(x-\Delta)/g(y)$ is not bounded by $\exp(\ep)$ almost everywhere, for all $\Delta \in [-1,1]$; by \citet[Proposition 2.3]{awan2019benefits}, this means that the additive mechanism with density $g$ does not satisfy $\ep$-DP. In other words, for $X\sim g$ and $\Delta$ as above, $T(X,X+\Delta)(\alpha)< f_{\ep,0}(\alpha)$ for some $\alpha\in (0,1)$. We conclude that $g$ is not a CND for $f_{\ep,0}$. The only assumption we made about $g$ was that it was non-constant on $(-1/2,1/2)$ on a set of positive probability. Due to the contradiction, we conclude that $g$ is in fact constant on $[-1/2,1/2]$ almost everywhere, which means that it is distributed as $\mathrm{Tulap}(0,\exp(-\ep),0)$. 
\end{proof}

\thmCNDlimit*
\begin{proof}
1) Let $f$ be a log-concave tradeoff function with log-concave CND $F$. Define $f_{s}(\alpha)=F(F^{-1}(\alpha)-s)$, which is a tradeoff function since $F$ is log-concave. Note that $f_s\circ f_t=f_{s+t}$, $f_t$ is nontrivial except wen $t=0$, and $f_t\rightarrow \mathrm{Id}$ as $t\rightarrow 0$. 

2) For part 2, we first show that the limit exists for the specific sequence $s_n = 1/n!$, and then we will show that convergence holds for all sequences that converge to zero. By construction, $F_{s_n}(\cdot)$ is a CND for $f_{s_n}$. So, $F_{s_n}(\cdot)$ has values determined on $(1/2)\ZZ    $, no matter the choice of CND. Furthermore, $F_{s_n}(t/s_n)$ is a CND for $f_1=f_{s_n}^{\circ n!}$, by Lemma \ref{lem:groupK}. Then for any choice of CND for $F_{s_n}$, the cdf $F_{s_n}(t/s_n)$ has values determined on $(s_n/2)\ZZ$, and $F_{s_n}(t/s_n)$ is a continuous cdf (as it is a CND). Note that the sequence $(s_n)_{n=1}^\infty$ satisfies $(s_n/2)\ZZ\subset (s_{n+1}/2)\ZZ$ for all $n$, and as $n\rightarrow \infty$, we have that  $\bigcup_{n=1}^\infty ((s_n/2)\ZZ)=\QQ$, the set of rational numbers. In words, the set of determined values of $F_{s_n}(t/s_n)$ is an increasing sequence of sets, whose limit is the rational numbers. Due to this, the sequence of CNDs $F_{s_n}(t/s_n)$ is ``pinned down'' at an increasing number of points, and is eventually determined at every rational number. Since every $F_{s_n}(t/s_n)$ is monotone, and $\QQ$ is dense in $\RR$, it follows that the limit of this sequence, $F^*$, is a unique monotone function.

   Next we show that $F^*$ is a continuous cdf. We already mentioned that $F^*$ is non-decreasing, and it is easy to show that $\lim_{t\rightarrow+\infty}F^*(t)=1$ and $\lim_{t\rightarrow -\infty} F^*(t)=0$. The challenging part is to show that $F^*$ is continuous. It suffices to show that the convergence of $(F_{s_n}(t/s_n))_{n=1}^\infty$ to $F^*$ is uniform. Before we show this, we establish the following inequality: for all $t\in \RR$,  $|F_{s_n}(t/s_n)-F^*(t)|\leq \sup_{t} |t-f_{s_n}(t)|$. To see this, let $t\in \RR$. Then for each $n$, there exists $k_n\in \ZZ$ such that $\frac{(k_n-1)s_n}{2}\leq t\leq \frac{(k_n+1)s_n}{2}$. Since $F^*$ is a non-decreasing function,  this implies that 
    \begin{align*}
        F^*\l(\frac{(k_n-1)s_n}{2}\r)&\leq F^*(t)\leq F^*\l(\frac{(k_n+1)s_n}{2}\r)\\
        F_{s_n}\l(\frac{k_n-1}{2}\r)&\leq F^*(t)\leq F_{s_n}\l(\frac{k_n+1}{2}\r)\\
        f_{s_n}\l(F_{s_n}\l(\frac{k_n+1}{2}\r)\r)&\leq F^*(t) \leq F_{s_n}\l(\frac{k_n+1}{2}\r),
    \end{align*}
    where if $F_{s_n}(\frac{k_n+1}{2})<1$ the third line is equivalent to the second line by Lemma \ref{lem:recurrence}, and if $F_{s_n}(\frac{k_n+1}{2})=1$, then the inequality in the third line is potentially weaker. 
    By similar reasoning, we have that $f_{s_n}(F_{s_n}(\frac{k_n+1}{2}))\leq F_{s_n}(t/s_n) \leq F_{s_n}(\frac{k_n+1}{2})$ as well. Therefore, 
    \[|F^*(t)-F_{s_n}(t/s_n)|\leq \l|F_{s_n}\l(\frac{k_n+1}{2}\r)-f_{s_n}\l(F_{s_n}\l(\frac{k_n+1}{2}\r)\r)\r|\leq \sup_{t\in [0,1]} |t-f_{s_n}(t)|.\]    
    
    We are now ready to prove uniform convergence. Let $\ep>0$ be given. Let $N\in \ZZ^+$ be such that $\sup_{t\in [0,1]} |t-f_{s_n}(t)|< \ep$, which is possible since $f_{s_n}(t)\rightarrow t$ uniformly (Polya's theorem). Then for all $n\geq N$, we have that $|F_{s_n}(t/s_n)-F^*(t)|\leq \sup_{t} |t-f_{s_n}(t)|<\ep$. Uniform convergence of continuous functions implies that the limit function is also continuous, so we conclude that $F^*$ is a continuous cdf.
    
Next we will show that for all $t\in \RR^+$, $f_t = F^*({F^*}^{-1}(\alpha)-t)$. Let $(-x,x)\defeq  {F^*}^{-1}((0,1))$ be the support of the distribution $F^*$. Let $q\in \QQ^+$ be the ``shift,'' and let $p\in (-\infty,x-q)\cap \QQ$ be the ``threshold'' in the test. Let $n\in \ZZ^+$ be such that $p=a/n!$ and $q=b/n!$ for some $a\in \ZZ$ and $b\in \ZZ^+$. As we did earlier, denote $s_n=1/n!$. Recall that $F^*(t)$ and $F_{s_n}(t/s_n)$ agree for all $t\in \frac{s_n}{2} \ZZ$. In particular, $p\in \frac{s_n}{2} \ZZ$. Then
\begin{align*}
    F^*(p)&=F^*(s_n a)\\
    &=F_{s_n}(a)\\
    &=f_{s_n}\circ F_{s_n}(a+1)\\
    &=f_{s_n}^{\circ b}\circ F_{s_n}(a+b)\\
    &=f_{q}\circ F_{s_n}(a+b)\\
    &=f_q \circ F^*(s_n(a+b))\\
    &=f_q\circ F^*(p+q),
\end{align*}
where we used the fact that $p<x-q$ to establish that $F_{s_n}(a+b)=F^*(p+q)<1$, which enabled the recurrence application of Lemma \ref{lem:recurrence}. 
Furthermore, since the rational numbers are dense in $\RR$, $F^*$ is continuous, and $f_q$ is continuous on $[0,1)$, we have that $F^*(t) = f_q\circ F^*(t+q)$ for all $t$ such that $t<x-q$. Now let $\alpha\in [0,1)$ and call $t_\alpha={F^*}^{-1}(\alpha)-q$. Note that $F^*(q+t_\alpha)=\alpha$ and that $t_\alpha<x-q$. Then we have that $F^*({F^*}^{-1}(\alpha)-q)=f_q(\alpha)$ for all $\alpha\in [0,1)$. Finally, we extend the result for arbitrary $f_r$, $r\in \RR^+$. Let $q_n\in \QQ^+$ be a sequence such that $q_n\rightarrow r$. Then
\[f_r=f_{r-q_n}\circ f_{q_n}=f_{r-q_n}\circ F^*({F^*}^{-1}(\alpha)-q_n)
\rightarrow \mathrm{Id}\circ F^*({F^*}^{-1}(\alpha)-r),\]
since $f_{r-q_n}$ converges to $\mathrm{Id}$ uniformly, and $F^*$ is continuous. We have that $F^*(s\cdot)$ is a CND for $f_s$: the symmetry of $F^*$ is obvious, and the fact that $F^*$ satisfies DP follows by the property that $f_s=F^*({F^*}^{-1}(\alpha)-s)$.

   
    
    Let $s>0$. Let $N\sim F^*$. Since $N/s\sim F^*(st)$ is a CND for $f_s$, we have that $T(N,N+s)(\alpha)=T(N/s,N/s+1)(\alpha)=f_s(\alpha)=F^*({F^*}^{-1}(\alpha)-s)$. However, $T(N,N+s)(\alpha)=F^*({F^*}^{-1}(\alpha)-s)$ holds for all $s>0$ if and only if $F^*$ has a log-concave density \citep[Lemma A.3]{dong2022gaussian}.
    Therefore $F^*$ is a log-concave distribution, and $F^*(s\cdot)$ is a log-concave CND for $f_s$ for all $s>0$. 
    

    

    Finally, we will make sure the limit does not depend on the specific sequence $s_n=1/n!$. We will use a very similar argument as when we established uniform convergence to show that for any positive sequence $r_n$ which converges to zero, $F_{r_n}(t/r_n)$ also converges uniformly to $F^*(t)$. Let $t\in \RR$. Then for any $n\in \ZZ^+$, there exists $k_n$ such that  $\frac{(k_n-1)r_n}{2}\leq t\leq \frac{(k+1)r_n}{2}$. Then 
    \[F^*\l(\frac{(k_n-1)r_n}{2}\r)\leq F^*(t)\leq F^*\l(\frac{(k_n+1)r_n}{2}\r).\]
    Since $F^*(\cdot r_n)$ and $F_{r_n}(\cdot)$ are both CNDs for $f_{r_n}$, they agree on all half integer values. So,
    \begin{align*}
        F_{r_n}\l(\frac{k_n-1}{2}\r)&\leq F^*(t)\leq F_{r_n}\l(\frac{k_n+1}{2}\r)\\
        f_{r_n}\l(F_{r_n}\l(\frac{k+1}{2}\r)\r)&\leq F^*(t)\leq F_{r_n}\l(\frac{k_n+1}{2}\r).
    \end{align*}
    By similar reasoning, we have that $f_{r_n}\l(F_{r_n}\l(\frac{k+1}{2}\r)\r)\leq F_{r_n}(t/r_n)\leq F_{r_n}\l(\frac{k_n+1}{2}\r)$. Then 
    \[|F^*(t)-F_{r_n}(t/r_n)|\leq \l|F_{r_n}\l(\frac{k_n+1}{2}\r)-f_{r_n}\l(F_{r_n}\l(\frac{k_n+1}{2}\r)\r)\r|\leq \sup_{t\in [0,1]} |t-f_{r_n}(t)|.\]
    Since $r_n\rightarrow 0$, we have that $f_{r_n}(t)$ converges uniformly to $t$. So, we have that $F_{r_n}(t/r_n)$ converges uniformly to $F^*(t)$.
    \end{proof}

\begin{restatable}{lem}{lempiecewise}\label{lem:piecewise}
Let $f$ and $g$ be tradeoff functions. 
\begin{enumerate}
    \item If $f$ and $g$ are piece-wise linear with $k$ and $\ell$ break points (respectively), and $f$ satisfies $f(x)=0$ implies $x=0$, then $f\circ g$ is piece-wise linear with at most $k+\ell$ break points, and at least $\max\{k,\ell\}$ break points.
     \item If $f$ is piece-wise linear with $k\geq 1$ break points, and  $f(x)=0$ implies that $x=0$, then $f^{\circ n}$ has at least $k+(n-1)$ break points.
    \item If $f\circ g$ is piece-wise linear, then $g$ is piece-wise linear on $[0,1]$ and $f$ is piece-wise linear on $[0,g(1)]$. (note that $f$ can be arbitrary on $(g(1),1]$ and it does not affect $f\circ g$)
\end{enumerate}
\end{restatable}

\begin{proof}
    \begin{enumerate}
    \item The composition of linear functions is linear. So, it is clear that $f\circ g$ is piece-wise linear. Let $B_g$ be the set of break points of $g$ and $B_f$ be the set of break points of $f$. Then the break points of $f\circ g$ are $g^{-1}(B_f)\bigcup B_g$, since $f$ is invertible, which has at most $|B_f|+|B_g|=k+\ell$ elements, and at least $\max\{k,\ell\}$ elements.
   \item Let $B_f$ be the set of break points of $f$. Then the set of break points of $f^{\circ n}$ is $B_f\cup f^{-1}(B_f)\cup f^{-1}(f^{-1})(B_f)\cup\cdots\cup (f^{-1})^{\circ (n-1)}(B_f)$. The number of break points of $f^{\circ n}$ is then lower bounded by $|B_f\cup f(B_f)\cup f^{\circ 2}(B_f)\cup\cdots \cup f^{\circ(n-1)}(B_f)|$, by applying $f^{\circ(n-1)}$ to each of the sets (since applying a function to a set cannot increase the cardinality). We know that $|B_f|=k$. Because  $f(x)=0$ implies that $x=0$, we have that $f$ is strictly increasing on $(0,1)$; so we have that $|f(B_f)|=k$ as well. Furthermore, for each $x\in B_f$, $f(x)<x$ as $f$ is nontrivial ($k\geq 1$ implies nontrivial). Let $x_m$ be the minimum element in $B_f$. Then $f(x_m)\in f(B_f)$ and $f(x_m)\not\in B_f$. So, $|B_f\cup f(B_f)|\geq |B_f|+1=k+1$. Continuing this process, we get that the number of break points of $f^{\circ n}$ is at least $k+n-1$. 
        \item  Since $f\circ g$ is piece-wise linear, $\frac{d^2}{d\alpha^2} f\circ g(\alpha)=0$ except at finitely many values. Then 
        \begin{align*}
            0&=\frac{d^2}{d\alpha^2} f\circ g(\alpha)\\
            &=\frac{d}{d\alpha}\left( f'(g(\alpha)) g'(\alpha)\right)\\
            &= f''(g(\alpha)) (g'(\alpha))^2 + f'(g(\alpha)) g''(\alpha),
        \end{align*}
        except at finitely many values. 
        Note that since $f$ and $g$ are non-decreasing and convex, the following quantities are non-negative (whenever they are well-defined): $f'$, $g'$, $f''$, and $g''$. So, the above equation implies that for all but finitely many $\alpha$, either $f''(g(\alpha))=0$ or $g'(\alpha)=0$ and either $f'(g(\alpha))=0$ or $g''(\alpha)=0$. 
        Note that $g$ is zero on $\{\alpha\mid g'(\alpha)=0\}$ and $f$ is zero on $\{g(\alpha)\mid f'(g(\alpha))=0\}$. Furthermore, $g$ is piece-wise linear on $\{\alpha\mid g''(\alpha)=0\}$ and $f$ is piece-wise linear on $\{g(\alpha)\mid f''(g(\alpha))=0\}$. We see that on $(0,1)$, $g$ is either zero or piece-wise linear, and so it is piece-wise linear on $[0,1]$. Similarly on $(g(0),g(1))=(0,g(1))$, $f$ is either zero or piece-wise linear, and so it is piece-wise linear on $[0,g(1)]$. 
        \end{enumerate} 
\end{proof}

\proppiecewise*
\begin{proof}
    \begin{enumerate}
        
        \item By part 3 of Lemma \ref{lem:piecewise}, if $f$ could be written as $g^{\circ (k+1)}=f$, then $g$ must also be piece-wise linear. Since $f(x)=0$ implies that $x=0$, $g$ also satisfy $g(x)=0$ implies $x=0$. If $g$ is a nontrivial piece-wise linear tradeoff function (with $j\geq 1$ breakpoints), then by part 2 of Lemma \ref{lem:piecewise}, $g^{\circ(k+1)}$ has $j+(k+1)-1>k$ break points. This contradicts that $g^{\circ(k+1)}=f$.
        \item Suppose that $f_{\ep,0}=g\circ h$, where both $g$ and $h$ are non-trivial. By part 3 of Lemma \ref{lem:piecewise}, we know that $h$ is piece-wise linear. Since we are assuming that $h$ is non-trivial, it must have at least one break point. Since $f_{\ep,0}(x)=0$ implies that $x=0$, $g$ must have this property as well. By Lemma part 1 of \ref{lem:piecewise}, this implies that $h$ must have a single break point. To agree with $f_{\ep,0}$, the breakpoint of $h$  must be at $1-c$, where $c=1/(1+\exp(\ep))$ is the solution to $f_{\ep,0}(1-c)=c$, since this is where the unique breakpoint of $f_{\ep,0}$ lies. However, since $h$ is a symmetric piece-wise linear function with a unique breakpoint at $1-c$, the only possibility is that $h=f_{\ep,0}$. 
        
        \item Each application of composition, functional composition, and subsampling without replacement preserves the piece-wise property of the tradeoff function, as well as the property that $f(x)=0$ implies that $x=0$. The result follows from part (a).
    \end{enumerate}
\end{proof}

\propCNDcomposition*
\begin{proof}
    For property 1 of Definition \ref{def:CND_MVT}, let $v$ be such that $\lVert v \rVert_\infty\leq 1$. Then $v=(v_1,\ldots, v_k)$ is such that $|v_i|\leq 1$. Then 
    \begin{align*}
        T(N,N+v)&=T((N_1,\ldots, N_k),(N_1+v_1,\ldots,N_k+v_k))\\
        &= T(N_1,N_1+v_1)\otimes \cdots \otimes T(N_k,N_k+v_k)\\
        &\geq f_1\otimes \cdots \otimes f_k.
    \end{align*}
    If we set $v^*_i=1$ for all $i$, then repeating the above gives equality in the last step, proving property 2 of Definition \ref{def:CND_MVT}. Call $g(x_1,\ldots, x_n)\defeq F'_1(x_1)\cdots F'_k(x_k)$ the density of $N$. For property 4 of Definition \ref{def:CND_MVT}, since $F'_i$ is symmetric about zero we have that $g$ is also symmetric about zero. For property 3, let $w=(w_1,\ldots, w_k)^\top$ be any vector and $v^*=(1,1,\ldots, 1)^\top$. Then, 
    \begin{align*}
        \frac{g(w+tv^*-v^*)}{g(w+tv^*)}
        &=\frac{F'_1(w_1+(t-1))\cdots F'_k(w_k+(t-1))}{F'_1(w_1+t)\cdots F'_k(w_k+t)},
    \end{align*}
    which is increasing in $t$, since each of the factors is increasing in $t$, by property 3 of Definition \ref{def:CND}.
    
\end{proof}

\thmlonemech*
\begin{proof}
    Since $f$ is a nontrivial log-concave tradeoff function, by Theorem \ref{thm:CNDlimit} there exists a monoid of log-concave tradeoff functions $\{f_t\in \mscr F\mid t\geq 0\}$ satisfying $f_t\circ f_s=f_{t+s}$ such that $f_1=f$ and $f_t=F(F^{-1}(\alpha)-t)$ for all $t>0$.  Note that for any $t\in \RR$, $T(N_i,N_i+t)=f_{|t|}$. 

    For property 1 of Definition \ref{def:CND_MVT}, let $x$ be such that $\lVert x\rVert_1\leq 1$. Note that $|x_i|<1$ for all $i=1,\ldots, k$. Then 
    \begin{align*}
        T(N,N+x)&=f_{|x_1|}\otimes f_{|x_2|}\otimes \cdots\otimes f_{|x_k|}\\
        &\geq f_{|x_1|}\circ f_{|x_2|}\circ \cdots\circ f_{|x_k|}\\
        &=f_{\sum_{i=1}^k |x_i|}\\
        &=f_{\lVert x_i\rVert_1}\\
        &\geq f_1,
    \end{align*}
    where in the first line, we use the property that $T(N_i,N_i+x_i)=f_{|x_i|}$ which uses log-concavity, the second line uses Lemma \ref{lem:otimesCirc}, and the third line uses the property that $f_s\circ f_t$ within the monoid. Note that for $v=(1,0,0,\ldots,0)$, $T(N,N+v)=T(N_1,N_1+1)=f$, proving property 2 of Definition \ref{def:CND_MVT}. Since $N$ is constructed by independent 1-d CNDs, the same arguments used in the proof of Proposition \ref{prop:CNDcomposition} can be used to prove properties 3 and 4 of Definition \ref{def:CND_MVT}. Note that $N$ is log-concave, since it is a product distribution with log-concave components.
\end{proof}

\propGauss*
\begin{proof}
Let $N\sim N(0, \Sigma)$. Note that $\Sigma^{-1/2}N \sim N(0, I_d)$. Let $u$ be such that $\lVert u \rVert\leq 1$. Then
    \begin{align*}
        T(N,N+u)&=T\left(\Sigma^{-1/2} N, \Sigma^{-1/2} N+\Sigma^{-1/2}u\right)\\
        &=T\left(N(0,I_d),N(0,I_d)+\Sigma^{-1/2} u\right)\\
        &=T\left(N(0,1), N(\lVert \Sigma^{-1/2}u\rVert_2,1)\right)\\
        &=G_{\lVert \Sigma^{-1/2} u\rVert_2}\\
        &\geq G_{\lVert \Sigma^{-1/2} v^*\rVert_2},
    \end{align*}
    where for the third line, we use the rotational invariance of the multivariate Gaussian distribution. Note that setting $u=v^*$ gives equality. This establishes properties 1 and 2 of Definition \ref{def:CND_MVT}. Property 4 of Definition \ref{def:CND_MVT} holds since the density of  $N(0, \Sigma)$ is symmetric about zero. For property 3, let $w\in \RR^k$ be any vector, and call $g$ the density of $N(0, \Sigma)$. Call $a=(w+tv^*)$ and $b=-v^*$. Then,  
    \begin{align*}
        \log\frac{g(w+tv^*-v^*)}{g(w+tv^*)}
        &=\log \frac{\exp(-\frac 12 (w+(t-1)v^*)^\top \Sigma^{-1} (w+(t-1)v^*))}{\exp(-\frac 12 (w+tv^*)^\top \Sigma^{-1} (w+tv^*))}\\
        &=-\frac 12 \left[(w+tv^*)-v^*)^\top \Sigma^{-1} ((w+tv^*)-v^*)+(w+tv^*)^\top \Sigma^{-1} (w+tv^*)\right]\\
        &= -\frac 12\left[(a+b)^\top \Sigma^{-1}(a+b)-a^\top \Sigma^{-1} a\right]\\
        &=-\frac 12 \left[{a^\top \Sigma^{-1} a}+2a^\top \Sigma^{-1} b+b^\top \Sigma^{-1} b-{a^\top \Sigma^{-1} a}\right]\\
        &=-\frac 12 \left[2 a^\top \Sigma^{-1} b+b^\top \Sigma^{-1} b\right]\\
        &=-\frac 12 \left[2(w+tv^*)^\top \Sigma^{-1} (-v^*)+(v^*)^\top \Sigma^{-1} v^*\right]\\
        &= t (v^*)^\top \Sigma^{-1} v^* + w^\top \Sigma^{-1} v^* - \frac 12 (v^*)^\top \Sigma^{-1} v^*,
    \end{align*}
    which is increasing in $t$, since $\Sigma$ is positive definite, which verifies property 3 of Definition \ref{def:CND_MVT}.
\end{proof}

\propCNDTVDP*
\begin{proof}
    Let $X\sim U(\frac{-1}{2\de},\frac{1}{2\de})^n$, and let $v$ be such that $\lVert v\rVert\leq 1$. We need a lower bound on $T(X,X+v)$. Since this is the testing of shifted uniforms, $T(X,X+v) = f_{0,\TV(X,X+v)}$. 
    \begin{align*}
        \TV(X,X+v)&= 1-\frac{\prod_{i=1}^d \left(\frac{1}{\de} - |v_i|\right)}{\de^{-d}}\\
        &=1-\prod_{i=1}^d \left(1 - \de|v_i|\right)\\
        &\geq 1-A,
    \end{align*}
    which establishes property 1 of Definition \ref{def:CND_MVT}. Note that using $v^*$ as defined above, we get that $\TV(X,X+v^*)=1-A$, giving property 2 of Definition \ref{def:CND_MVT}. Property 4 of Definition \ref{def:CND_MVT} is obvious, since each uniform is centered at zero. For property 3, let $w\in \RR^d$. The likelihood ratio is 
    \begin{align*}
        \frac{g(w+(t-1)v^*)}{g(w+tv^*)}
        &=\prod_{i=1}^d \frac{I(\frac{-1}{2\delta}\leq w_i+(t-1)v_i^*\leq \frac{1}{2\delta})}{I(\frac{-1}{2\delta} \leq w_i+tv_i^*\leq \frac{1}{2\delta})},
    \end{align*}
    and we see that each of these factors can take the possible values:
    \[\begin{cases}
    \text{undefined} & \text{when }w_i+(t-1)v_i^*\not\in [\frac{-1}{2\delta},\frac{1}{2\delta}] \text{ and } w_i+tv_i^*\not\in [\frac{-1}{2\delta},\frac{1}{2\delta}]\\
        0 &\text{when }w_i+(t-1) v_i^*\not\in [\frac{-1}{2\delta},\frac{1}{2\delta}] \text{ and } w_i+tv_i^*\in [\frac{-1}{2\delta},\frac{1}{2\delta}]\\
        1 &\text{when }w_i+(t-1) v_i^*\in [\frac{-1}{2\delta},\frac{1}{2\delta}] \text{ and } w_i+tv_i^*\in [\frac{-1}{2\delta},\frac{1}{2\delta}]\\
        +\infty &\text{when }w_i+(t-1) v_i^*\in [\frac{-1}{2\delta},\frac{1}{2\delta}] \text{ and } w_i+tv_i^*\not\in [\frac{-1}{2\delta},\frac{1}{2\delta}]\\
    \end{cases}\]
    If $w_i+tv_i^*\in [\frac{-1}{2\delta},\frac{1}{2\delta}]$ for some $t$, then we have that as $t$ progresses from $-\infty$ to $\infty$, the value of each factor goes from undefined to $0$ to $1$ to $+\infty$ to undefined, which is a non-decreasing sequence. If $w_i+tv_i^*\not\in [\frac{-1}{2\delta},\frac{1}{2\delta}]$ for every $t$, then the likelihood ratio is always undefined, which is also trivially non-decreasing. We see that property 3 of Definition \ref{def:CND_MVT} holds.    
\end{proof}

The privacy loss random variable is a concept that appears in all major definitions of differential privacy. In fact, \citet{zhu2022optimal} showed that the privacy loss random variables can be losslessly converted back and forth to the corresponding tradeoff function. For part of the proof of Proposition \ref{prop:Linfty}, it will be easier to work with the privacy loss random variables than directly with the tradeoff functions. First, we give a formal definition and a few basic properties of privacy loss random variables. While similar results appeared in \citet{zhu2022optimal}, we include them here for completeness.

\begin{defn}[Privacy Loss Random Variable]
Let $X$ and $Y$ be two random variables on $\RR^d$, with densities $p$ and $q$, respectively. The \emph{privacy loss random variable} is $\mathrm{PLRV}(X|Y) \defeq \log\frac{p(X)}{q(X)}$, where $X\sim p$.
\end{defn}

\begin{lemma}[Privacy Loss RV is Sufficient]\label{lem:PLsuff}
Let $X\sim p$ and $Y\sim q$ be two random variables on $\RR^d$ with densities $p$ and $q$, respectively. Define $L(x):\mscr X\rightarrow \RR$ by $L(x) = \log[q(x)/p(x)]$. Note that $L(X) \overset d= -\mathrm{PLRV}(X|Y)$ and $L(Y) \overset d= \mathrm{PLRV}(Y|X)$. Then $T(X,Y)=T(L(X),L(Y))=T(-\mathrm{PLRV}(X|Y),\mathrm{PLRV}(Y|X))$. 
\end{lemma}
\begin{proof}
First, by postprocessing, we have that $T(X,Y)\leq T(L(X),L(Y))$ \citep[Lemma 2.9]{dong2022gaussian}. For the other direction, note that by the Neyman Pearson Lemma, the optimal test for $H_0: X$ versus $H_1: Y$ at size $\alpha$ is of the form 
    \[\phi(x) = \begin{cases}
    1 & L(x)>t\\
    c& L(x)=t\\
    0&L(x)<t,
    \end{cases}\]
    where $L$ is defined in the Lemma statement, and the values of $t$ and $c$ are uniquely chosen such that $\EE_{X\sim p}\phi(X)=\alpha$. 
    
    For a specified $t$ and $c$, the type I error is 
    \begin{align*}
    \text{type I} = \EE_{X\sim p}\phi &= \EE_{X\sim p} [I(L(X)>t)+cI(L(X)=t)]\\
    &=P_{X\sim p}(L(X)>t)+cP_{X\sim p}(L(X)=t),
    \end{align*}
    which we see only depends on the statistic $L(X)$. On the other hand,
    
    \begin{align*}
        \text{type II} = 1-\EE_{Y\sim q}\phi(Q) &= 1-\EE_{Y\sim q}[I(L(Y)>t)+cI(L(Y)=t)]\\
    &=P_{Y\sim q}(L(Y)\leq t)-cP_{Y\sim q}(L(Y)=t),
    \end{align*}
    which we see only depends on the statistic $L(Y)$.
    
    So, when testing $H_0: L(X)$ versus $H_1: L(Y)$, using the particular test $\psi(L) = I(L>t)+cI(L=t)$,  
    where the values of $c$ and $t$ are chosen as above, we recover the type I and type II errors of $T(X,Y)$. We conclude that $T(L(X),L(Y))\leq T(X,Y)$. Combining the inequalities, we have that $T(X,Y) = T(L(X),L(Y))$. The fact that $L(X) = -\mathrm{PLRV}(X|Y)$ and $L(Y) = \mathrm{PLRV}(Y|X)$ follows from the definition of privacy loss random variables.
\end{proof}

\begin{lemma}\label{lem:matchPLRV}
Let $X\in \RR^d$ be a continuous random vector with density $g$, which is symmetric about zero. Then for any $v\in \RR^d$, $\mathrm{PLRV}(X|X+v) = \mathrm{PLRV}(X+v|X)$. It follows that 
\begin{enumerate}
    \item $T(X,X+v) = T(\mathrm{PLRV(X|X+v)},-\mathrm{PLRV}(X|X+v))$, and 
    \item Let $Y\in \RR^p$ be another continuous random vector symmetric about zero, and let $w\in \RR^p$. Then if $\mathrm{PLRV}(X|X+v)\overset d =\mathrm{PLRV}(Y|Y+w)$ then $T(X,X+v) = T(Y,Y+w)$.
\end{enumerate}
\end{lemma}
\begin{proof}
First note that $\mathrm{PLRV}(X|X+v) = \log \frac{g(X)}{g(X-v)}$, where $X\sim g$. Setting $Z=v-X$, we can write 
\begin{align*}
    \mathrm{PLRV}(X+v|X) &= \log \frac{g(X-v)}{g(X)}, \text{ where } (X-v)\sim g\\
    &=\log\frac{g(-Z)}{g(-Z+v)},\text{ where } -Z\sim g\\
     \quad &=\log\frac{g(Z)}{g(Z-v)}, \text{ where } Z\sim g\quad \text{(by symmetry of $g$)}\\
    &\overset d = \mathrm{PLRV}(X|X+v).
\end{align*}
Combining the above work with Lemma \ref{lem:PLsuff}, we get $T(-\mathrm{PLRV}(X|X+v),\mathrm{PLRV}(X|X+v))$, which is equivalent to statement 1, since the tradeoff function is symmetric. For part two, if $\mathrm{PLRV}(X|X+v)\overset d =\mathrm{PLRV}(Y|Y+w)$, then clearly $T(\mathrm{PLRV}(X|X+v),-\mathrm{PLRV}(X|X+v))=T(\mathrm{PLRV}(Y|Y+w),-\mathrm{PLRV}(Y|Y+w))$, which is equivalent to the statement in part 2, by part 1.
\end{proof}

\propLinfty*
\begin{proof}
Note that for any vector $s\in \{-1,1\}^d$, $sX\overset d = X$ (entry-wise multiplication). Because of this, it suffices to consider $T(X,X+v)$ for $v\geq 0$ (all entries non-negative).

First we will show that $T(X,X+1)=T(L,L+1)$, where $X$ is the $\ell_\infty$-mech, and $L\sim \mathrm{Laplace}(0,1/\ep)$ which has density $\frac{\ep}{2}\exp(-\ep|x|)$. We will do this using privacy loss random variables, applying part 2 of Lemma \ref{lem:matchPLRV}. Note that since $X$ and $L$ are both symmetric random variables, it suffices to equate the privacy loss random variables $\mathrm{PLRV}(X|X+1)$ and $\mathrm{PLRV}(L|L+1)$. We can easily derive that $\mathrm{PLRV}(L|L+1)= -\ep|L|+\ep|L-1|=\ep[1-2L]_{-1}^{1}$, where $L\sim \mathrm{Laplace}(0,1/\ep)$ and $[x]_a^b\defeq \min\{\max\{x,a\},b\}$ is the clamping function. Note that $\mathrm{PLRV}(L|L+1)\overset d =\ep[1-L_2]_{-1}^1$, where $L_2\sim \mathrm{Laplace}(2/\ep)$. 

Now for $T(X,X+1)$, the privacy loss random variable is $\mathrm{PLRV}(X|X+1)= -\ep\lVert X\rVert_\infty+\ep\lVert X-1\rVert_\infty$, where $X\sim g(x)$. We can simplify this expression as follows, using the notation $\max(X)=\max_{1\leq i\leq d} X_i$ and $\min(X)=\min_{1\leq i\leq d} X_i$:

\begin{align*}
    &\mathrm{PLRV}(X|X+1)\\
    &=-\ep\lVert X\rVert_\infty+\ep\lVert X-1\rVert_\infty\\
    &=\begin{cases}
    -\ep\max(X)+\ep(1-\min(X)) & \text{if } \max(X)\geq -\min(X)\ \&\  1-\min(X)\geq \max(X)-1\\
    -\ep\max(X)+\ep(\max(X)-1) & \text{if } \max(X)-1>1-\min(X) \\
    -\ep(-\min(X))+\ep(1-\min(X))&\text{if } -\min(X)>\max(X) \\
    \end{cases}\\
    &=\begin{cases}
    -\ep\max(X)+\ep(1-\min(X)) & \text{if } -1\leq [1-(\max(X)+\min(X))]\leq 1 \\
    -\ep & \text{if } [1-(\max(X)+\min(X))]<-1\\
    \ep &\text{if } [1-(\max(X)+\min(X))]>1 \\
    \end{cases}\\
    &=\ep[1-(\max(X)+\min(X))]_{-1}^1.
\end{align*}

Comparing this expression with $\mathrm{PLRV}(L|L+1)$, we see that it suffices to show $\max(X)+\min(X) \overset d = \mathrm{Laplace}(2/\ep)$. Recall that $X\overset d=RU$, where $R\sim \mathrm{Gamma}(d+1,\ep)$, using the shape, rate parameterization, and $U_i\iid U(-1,1)$ for $i=1,\ldots, d$ \citep[Remark 4.2]{hardt2010geometry}. By factoring out $R$, we get 
\[\max(X)+\min(X)\overset d= R(\max(U)+\min(U)).\]
So, we will determine the distribution of $\max(U)+\min(U)$ first. We can easily compute the joint distribution of $\max(U)$ and $\min(U)$, as these are the minimum and maximum order statistics:
\[f_{\min(U),\max(U)}(x,y) = \frac{d(d-1)}{4} \left( \frac{y-x}{2}\right)^{d-2} I(-1\leq x\leq y\leq 1).\]
Now consider the change of variables $m=x$ and $w=x+y$. Applying change of variables, we have 
\begin{align*}
    f_{\min(U),\max(U)+\min(U)}(m,w) &= \frac{d(d-1)}{4} \left( \frac{(w-m)-m}{2}\right)^{d-2} I(-1\leq m\leq w-m\leq 1)\\
    &= \frac{d(d-1)}{2^{d}}\left(w-2m\right)^{d-2}I(-1\leq m, m\leq w/2, w-1\leq m).
\end{align*}
To get the distribution of $W=\max(U)+\min(U)$, we marginalize out $M=\min(U)$:
\begin{align*}
    f_{\max(U)+\min(U)}(w) &= \int_{\max\{-1,w-1\}}^{w/2} \frac{d(d-1)}{2^d} \left(w-2m\right)^{d-2} \ dm\\
    &=\frac{d(d-1)}{2^{d+1}} \frac{-(w-2m)^{d-1}}{d-1}\Big|_{\max\{-1,w-1\}}^{w/2}\\
    &= \frac{d}{2^{d+1}} \left[ (w-2\max\{-1,w-1\})^{d-1} - \cancel{(w-2(w/2))^{d-1}}\right]\\
    &= \frac{d}{2^{d+1}} \begin{cases}
    (w+2)^{d-1}&-2\leq w\leq 0\\
    (2-w)^{d-1}&0\leq w\leq 2\\
    \end{cases}\\
    &= \frac{d}{2^{d+1}}(2-|w|)^{d-1}I(-2\leq w\leq 2).
\end{align*}
Since the distribution of $W=\max(U)+\min(U)$ is symmetric about zero, $W \overset d= (-1)^B|W|$ where $B\sim \mathrm{Bern}(1/2)$. So, our goal is to show $(-1)^B RW \overset d = (-1)^B \mathrm{Exp}(\ep/2)$, since the left side is equal in distribution to $\max(X)+\min(X)$ and the right side is equal in distribution to $\mathrm{Laplace}(2/\ep)$. The pdf of $Y\overset d\defeq|W|$ is $f(y) = \frac{d}{2^d} (2-y)^{d-1}I(0\leq y\leq 2)$. It suffices to show that $RY \overset d = \mathrm{Exp}(\ep/2)$. Let $\phi_R$ be the characteristic function of $R$ and $\phi_{RY}$ be the characteristic function of $RY$. Then, 

\begin{align*}
    \phi_{RY}(t)&=\EE_Y \phi_R(tY)\\
    &=\EE_Y\left(1-\frac{ity}{\ep}\right)^{-(d+1)}\\
    &=\frac{d}{2^d} \int_0^2 \left(1-\frac{ity}{\ep}\right)^{-(d+1)} (2-y)^{d-1}\ dy\\
    &= \frac{\ep}{\ep-2it}\\
    &=\frac{\ep/2}{\ep/2-it},
\end{align*}
which we identify as the characteristic function of $\mathrm{Exp}(\ep/2)$, establishing that $RY \overset d= \mathrm{Exp}(\ep/2)$.  By part 2 of Lemma \ref{lem:matchPLRV}, this completes the proof that $T(X,X+1)=T(L,L+1)$, establishing property 2 of Definition \ref{def:CND_MVT}. Note that property 4 of Definition \ref{def:CND_MVT} is obvious, and property 3 holds since the likelihood ratio $g(x-1)/g(x)$ is an increasing function in $\max(x)+\min(x)$, which itself is an increasing function of $t$ when $x=w+t$ for every vector $w\in \RR^d$. It remains to verify property 1 of Definition \ref{def:CND_MVT}.

Next we will show that for $v\in (0,1]^d$, $T(X,X+1)\leq T(X,X+v)$ (this proof strategy is based on the proof of Lemma 3.5 in \citet{dong2021central}). We will separately address the cases that some of $v_i=0$ at the end of the proof. Call $f_1=T(X,X+1)$ and $f_v=T(X,X+v)=T(X/v,X/v+1)$. Define the two linear maps $r: x\mapsto x/v$ and $r^{-1}: x\mapsto vx$ (entry-wise multiplication and division), which are inverse maps. Note that $r(X)=X/v$ has density proportional to $\exp(-\ep\lVert vt\rVert_\infty)$. Let $\alpha\in[0,1]$ be given. Let $A$ be the optimal rejection region for $T(X/v,X/v+1)$ at type I error $\alpha$. By our earlier work, we know that 
\[A = \{x\mid \epsilon(\max(vx)+\min(vx))\geq t\},\]
for some $t\in \RR$, and it satisfies $P(r(X)\in A)=\alpha$ and $P(r(X)+1\not\in A)=f_v(1-\alpha)$. We can now consider $r^{-1}(A)$ as a possible rejection region for testing $T(X,X+1)$, which is at best suboptimal. We compute the type I error as 
\[P(X\in r^{-1}(A_v)) = P(r(X)\in A)=\alpha.\]
Suboptimality of the rejection region implies that 
\begin{align*}
    f_1(1-\alpha)&\leq P(X+1\not\in r^{-1}(A))\\
    &=P(r(X)+r(1)\not \in A)\\
    &=P(r(X)+1/v\not\in A)\\
    &=P(r(X)\not \in A-1/v)\\
    &\leq P(r(X)\not\in A-1)\\
    &=f_v(1-\alpha),
\end{align*}
where we used the fact that $r$ is a linear map, and $r(1)=1/v$; the key step is the final inequality, which we justify as follows: it suffices to show that $(A-1/v)^c\subset (A-1)^c$ or equivalently $A-(1/v-1)\supset A$. We verify this by inspecting the definition of A:
\begin{align*}
    A&=\{x\mid \max(vx)+\min(vx)\geq t\}\\
    &\subset \{x\mid \max(vx+(1/v-1))+\min(vx+(1/v-1))\geq t\}\\
    &=\{x-(1/v-1)\mid \max(vx)+\min(vx)\geq t\}\\
    &= A-(1/v-1),
\end{align*}
where in the inclusion step, we used the fact that $(1/v-1)\geq 0$ implies that $\max(vx+(1/v-1))\geq \max(vx)$ and $\min(vx+(1/v-1))\geq \min(vx)$. This completes the argument that for $v\in (0,1]^d$, $f_1=T(X,X+1)\leq T(X,X+v)=f_v$.

Finally, let $v\in [0,1]^d$, where the entries may possibly be zero. Let $v_n\in (0,1]^d$ be a sequence of vectors converging to $v$. Notice that $X+v_n\overset {TV}\rightarrow X+v$, since $X$ has a continuous density. Since $T(X,X+v_n)\geq T(X,X+1)$ by our above work, by Corollary \ref{cor:limitBound} we have $T(X,X+v)\geq T(X,X+1)$ as well. 
\end{proof}

\thmnoCNDpure*

The proof strategy of Theorem \ref{thm:noCNDpure} is as follows: 1) observe that property of Definition \ref{def:CND_MVT} enforces constraints on the likelihood ratio $\log \frac{g(x-v)}{g(x)}$, 2) establish that the measure induced by $g$ is equivalent to Lebesgue measure, which simplifies some measure theory details, 3) show that we can construct a vector $w$ such that $\lVert w \rVert<1$, $\lVert w+v\rVert<1$, and $w\not\in \mathrm{Span}(v)$, 4) based on the properties of $w$ and $v$, show that by taking integral combinations of $w$ and $v$, we can find an arbitrarily long sequence of points each sufficiently far from each other such that the value of $g$ is bounded below by a common constant, and 5) show that point 4 implies that $g$ is not integrable. Because densities are only well defined up to sets of Lebesgue measure zero, the details of the proof are more complicated to ensure that we are careful about the measure theoretical details. 

\begin{proof}
 Suppose to the contrary that there exists a CND for $f_{\ep,0}$ with respect to $\lVert \cdot \rVert$, which has density $g$. We will denote $\mu_g$ as the measure induced by $g$: $\mu_g(S) = \int_S g(x) \ dx$, and use $\lambda$ to denote Lebesgue measure.

 By property 2 of Definition \ref{def:CND_MVT}, there exists $v\in \RR^d$ such that $\lVert v \rVert\leq 1$ and $T(N,v+N)=f_{\ep,0}$, where $N\sim g$. This implies that $\log \frac{g(x-v)}{g(x)}=\pm \ep$  almost everywhere ($\mu_g$) for all $x\in \RR^d$ (if $P$ and $Q$ are two distributions satisfying $T(P,Q)=f_{\ep,0}$, then the privacy loss random variable is a binary random variable, taking values in $\{-\ep,\ep\}$). Furthermore, by property 1, for any other vector $w\in \RR^d$ such that $\lVert w \rVert\leq 1$, we have that $\log \frac{g(x-w)}{g(x)}\in [-\ep,\ep]$  for almost every $x\in \RR^d$ ($\mu_g$). 

Before we begin our main argument, we will show that (if such a $g$ exists,) $\mu_g$ must be equivalent to Lebesgue measure. We know that Lebesgue measure dominates $\mu_g$, so we only need to show that $\mu_g=0$ implies $\lambda=0$. Suppose to the contrary that there exists $S\subset \RR^d$ such that $\lambda(S)>0$ but $\mu_g(S)=0$ (which implies that $g(x)=0$ a.e. on $S$). We claim that there exists such an $S$ such that for some $\lVert t\rVert\leq 1$, $\mu_g(S+t)>0$. We prove this as follows: begin with any $S$ such that $\lambda(S)>0$ but $\mu_g(S)=0$. If $\mu_g(S+t)=0$ for all $\lVert t \rVert\leq 1$, then set $S'=\bigcup_{\lVert t \rVert\leq 1} (S+t)$, which is strictly larger than $S$. If $S'$ still does not have the desired property, repeat the process iteratively. Note that in the limit, this process results in $\RR^d$, but $\mu_g(\RR^d)=1$. So, the process must terminate, giving us the desired set with the properties $\lambda(S)>0$, $\mu_g(S)=0$ and there exists some $\lVert t\rVert\leq 1$ such that $\mu_g(S+t)>0$. Then there exists $P\subset S+t$ such that $g(x)>0$ on $P$ a.e., and note that $P-t\subset S$ and $g(x)=0$ on $P-t$ a.e.. However, this implies that $\log \frac{g(x-t)}{g(x)}=\infty\not\in [-\ep,\ep]$ on the set $P$, which has positive probability $\mu_g(P)>0$. This contradicts property 1 of Definition \ref{def:CND_MVT}, as discussed above. We conclude that $\mu_g$ and $\lambda$ are equivalent measures. So, we will interchangeably use statements about Lebesgue measure and $\mu_g$ measure.

    
  Let $r>0$ be the largest value such that $\lVert x \rVert_2\leq r$ implies that $\lVert x \rVert\leq 1$ (possible by the equivalence of norms on $\RR^d$). Consider three sets 
    \begin{align*}
        A &= \{w\mid \lVert \mathrm{Proj}_{v^\perp} w\rVert_2<r/4 \ \& \ \lVert w \rVert < 1\},\\
        B &= \{w\mid \lVert \mathrm{Proj}_{v^\perp} w \rVert_2\leq r/8 \ \& \ \lVert w \rVert<1\},\\
        C &= \{w\mid \lVert w+v\rVert<1\}.
    \end{align*}
    Note that $(A\setminus B)\cap C$ is an open set; we will demonstrate that it is non-empty, which implies that it has non-zero Lebesgue measure. First note that $A\setminus B\neq\emptyset$, since $d\geq 2$ implies that $\emptyset \subsetneq B\subsetneq A$. We will construct a vector $w\in (A\setminus B)\cap C$ as follows: let $y\in A\setminus B$, and call $z=\mathrm{Proj}_{v^\perp} y$. Then 
    $r/8<\lVert z\rVert_2<r/4$. We set $w = z-v/2$. First we will check that $w\in A\setminus B$: since $z=\mathrm{Proj}_{v^\perp}w = \mathrm{Proj}_{v^\perp}y$, we have that $r/8<\lVert \mathrm{Proj}_{v^\perp}w \rVert<r/4$. We also need to check that 
    \[\lVert w \rVert=\lVert z-v/2\rVert\leq \lVert z\rVert+\frac 12 \lVert v \rVert\leq \frac 14+\frac 12<1,\]
    since $\lVert z\rVert_2\leq r/4$ implies that $\lVert z\rVert\leq 1/4$. 
    Next, we will check that $w\in C$: 
    \[\lVert w+v\rVert=\lVert z-v/2+v\rVert=\lVert z+v/2\rVert \leq \lVert z\rVert+\frac12 \lVert v\rVert\leq \frac14+\frac12<1,\]
    using again the fact that $\lVert z\rVert_2\leq r/4$ implies that $\lVert z\rVert\leq 1/4$. We conclude that $(A\setminus B)\cap C$ is a non-empty open set, which implies that it has non-zero Lebesgue measure.
    
    Let $c\defeq\esssup g$. Note that $c<\infty$, as otherwise, this would violate the log-likelihood ratio property discussed earlier. Then $g(x)\leq c$ holds with probability one. So, since $(A\setminus B)\cap C$ has positive probability, we can find a vector $w\in (A\setminus B)\cap C$ which satisfies $g(w)\leq c$. 
    
    Let $\eta\in (0,c)$ be given. Then the set $\{x\mid c-\eta\leq g(x)\leq c\}$ has positive measure. For $v$ as defined above, and an arbitrary vector $u\in \RR^d$, consider two more sets:
    \begin{align*}G_v &= \left\{x\middle| \log \frac{ g(x+v)}{g(x)} \in \{-\ep,\ep\}\right\},\\
    F_u &= \left\{u\middle| \log \frac{g(x+u)}{g(x)} \in [-\ep,\ep]\right\},
    \end{align*}
        which both hold with probability one whenever $\lVert u \rVert\leq 1$.

    Let $K$ be a positive integer such that $K>\left( e^{-\ep} (c-\eta)\frac{\pi^{d/2}}{\Gamma(1+d/2)}(r/8)^d\right)$.  Then there exists $\xi\in \RR^d$ such that for every $0\leq j\leq k\leq K$,
    \begin{align*}
        \xi&\in G_v-(kw+jv),\\
        \xi&\in F_w -(kw+jv),\\
        \xi&\in F_{w+v}-(kw+jv),\\
        c-\eta&\leq g(\xi)\leq c,\\
        \lVert\xi\rVert_2&\leq b\defeq\esssup\limits_{c-\eta\leq g(x)\leq c}\lVert x \rVert_2,
    \end{align*}
    since the first three lines hold with probability one, the fourth holds with positive probability as discussed earlier, and the last holds with probability one. Note that $b<\infty$ as otherwise, we would have an unbounded region with positive probability such that $g\geq c-\eta>0$, which would imply that $g$ is not integrable.

    
    Since $\xi\in F_w$, we have that $g(\xi+w) \in [e^{-\ep}g(\xi),c]\subset  [e^{-\ep}(c-\eta),c]$, since $\lVert w\rVert\leq 1$. Similarly, as $\xi\in F_{w+v}$,  we have that $g(\xi+w+v)\in [e^{-\ep}(c-\eta),c]$, since $\lVert w+v\rVert\leq 1$. However, since $\xi \in G_v-w$, we have that $\frac{g(\xi+w+v)}{g(\xi+w)}=e^{\pm\ep}$. The only possibility to satisfy all of these constraints is for either $g(\xi+w)\geq c-\eta$  or for $g(\xi+w+v)\geq c-\eta$. 
    Repeating the previous argument, starting with $g(\xi+w)\geq c-\eta$ gives either $g(\xi+2w)\geq c-\eta$ or $g(\xi+2w+v)\geq c-\eta$. If instead, we start with $g(\xi+w+v)\geq c-\eta$, then either $g(\xi+2w+v)\geq c-\eta$ or $g(\xi+2w+2v)\geq c-\eta$. We see that after $k$ steps of this procedure, we have that $g(\xi+kw+jv)\geq c-\eta$ for some $j\in \{0,1,2,\ldots, k\}$. 
    We denote by $j(k)$ the value of $j$ obtained by this procedure at the $k^{th}$ step. 
    
    For each $0\leq k\leq K$, define
    \[A_k = \{x\mid \lVert x-(\xi+kw+j(k)v)\rVert_2< r/8\}.\]
    Note that since $w\in A\setminus B$ from above, we know that $\lVert \mathrm{Proj}_{v^\perp} w \rVert_2\geq r/8$. This implies that each $A_k$ is disjoint from the others, since the $A_k$ are of radius $r/8$, and the distance between each set is at least $r/8$. Furthermore, notice that on each $A_k$, $g\geq e^{-\ep}(c-\eta)$.
    
    Finally, consider the integral of $g$, which we lower bound:
    
    \begin{align*}
        \int_{\RR^d}g(x) \ dx&\geq \sum_{k=0}^K \int _{A_k} g(x) \ dx\\
        &\geq \sum_{k=0}^K \int_{A_k} e^{-\ep}(c-\eta) \ dx\\
        &=\sum_{k=0}^K e^{-\ep}(c-\eta) \mathrm{Vol}(A_k)\\
        &=(K+1)  e^{-\ep}(c-\eta) \frac{\pi^{d/2}}{\Gamma(1+d/2)} (r/8)^d\\
        &>1,
    \end{align*}
    where we used the formula for a $d$-dimensional sphere of radius $r/8$ to evaluate $\mathrm{Vol}(A_k)$, and in the last line, we used the fact that $K>\left( e^{-\ep} (c-\eta)\frac{\pi^{d/2}}{\Gamma(1+d/2)}(r/8)^d\right)$. We see that $g$ cannot integrate to one, which contradicts our assumption that it is a multivariate CND. In fact, $g$ is not even integrable, as $K$ could have been chosen arbitrarily high.

\end{proof}

\corpureDecomp*
\begin{proof}
Suppose to the contrary that there did exist a nontrivial decomposition $f_{\ep,0}=f_1\otimes f_2$. Then Proposition \ref{prop:CNDcomposition} gives a construction for a 2-dimensional CND of $f_{\ep,0}$. However, we know from Proposition \ref{thm:noCNDpure} that there is no 2-dimensional CND for $f_{\ep,0}$. 
\end{proof}
\end{document}